\documentclass[final]{siamltex}

\usepackage{amssymb,latexsym,amsmath}
\usepackage{epsfig}
\usepackage{caption}
\usepackage{caption}
 
 \usepackage{epsfig}

\usepackage{makeidx}

\allowdisplaybreaks

 \usepackage{graphicx,fancybox,latexsym,epsfig}
\usepackage{fancyhdr,amsmath,times,amsxtra,amssymb}
\usepackage{color}

\def\Pr{\mathop{\rm Pr}}

\def\B{{\mathcal B}}

\def\P{{\mathcal P}}

\def\sPr{{\mathsf{Pr}}}

\def\sX{{\mathds X}}
\def\sY{{\mathds Y}}
\def\sA{{\mathds A}}
\def\sU{{\mathds U}}
\def\sH{{\mathds H}}
\def\sZ{{\mathds Z}}

\def\sU{{\mathds U}}

\usepackage{amsmath}
\usepackage{amsfonts}
\usepackage{array}
\usepackage{dsfont}
\usepackage{hyperref}
\usepackage{amssymb}
\usepackage{bbold}
\usepackage{caption}
\usepackage{standalone}
\usepackage{accents}
\usepackage{enumitem}
\usepackage{tikz}
\usepackage{pgfplots}
\pgfplotsset{compat=1.6}
\usepgfplotslibrary{groupplots}
\usepackage{cancel}
\allowdisplaybreaks

\newtheorem{prop}{Proposition}[section]
\newtheorem{assumption}{Assumption}[section]
\newtheorem{exmp}{Example}[section]

 \newtheorem{remark}{Remark}

\usepackage{amssymb,latexsym,amsmath}
\usepackage{mathrsfs}
\usepackage{epsfig}
\usepackage{caption}

\newcommand{\R}{\mathds{R}}
\newcommand{\Zplus}{\mathds{Z}_+}
\newcommand{\N}{\mathds{N}}

\newcommand{\dd}{\mathrm{d}}

\newcommand{\adk}[1]{{\color{black} #1}}

\pgfplotsset{soldot/.style={color=blue,only marks,mark=*}}
\pgfplotsset{holdot/.style={color=blue,fill=white,only marks,mark=*}}
\fussy

\allowdisplaybreaks

\begin{document}
\sloppy
\title{Robustness to Incorrect Priors in Partially Observed Stochastic Control
\thanks{
This research was supported in part by
the Natural Sciences and Engineering Research Council (NSERC) of Canada. Part of this paper has been submitted to the 2018 IEEE Conference on Decision and Control in an abbreviated form.}
}
\author{Al\.{I} Devran Kara and Serdar Y\"uksel
\thanks{The authors are with the Department of Mathematics and Statistics,
     Queen's University, Kingston, ON, Canada,
     Email: \{16adk,yuksel\}@queensu.ca}
     }
\maketitle
\begin{abstract}
We study the continuity properties of optimal solutions to stochastic control problems with respect to initial probability measures and applications of these to the robustness of optimal control policies applied to systems with incomplete or incorrect priors. It is shown that for single and multi-stage optimal cost problems, continuity and robustness cannot be established under weak convergence in general, but that the optimal expected cost is continuous in the priors under the convergence in total variation under mild conditions. By imposing further assumptions on the measurement models, robustness and continuity also hold under weak convergence of priors. We thus obtain robustness results and bounds on the mismatch error that occurs due to the application of a control policy which is designed for an incorrectly estimated prior as the incorrect prior converges to the true one. Positive and negative practical implications of these results in empirical learning for stochastic control are presented, where almost surely weak convergence of i.i.d. empirical measures occurs but stronger notions of convergence, such as total variation convergence, in general, do not. \end{abstract}


\section{Introduction}\label{section:intro}
\subsection{Preliminaries}

We start with the probabilistic setup of the problem. Let $\mathds{X} \subset \mathds{R}^n$, be a Borel set in which elements of a
controlled Markov process $\{X_t,\, t \in \Zplus\}$ live.  Here
and throughout the paper $\Zplus$ denotes the set of non-negative
integers and $\mathds{N}$ denotes the set of positive integers.  Let
$\mathds{Y} \subset \mathds{R}^m$ be a Borel set, and let an
observation channel $Q$ be defined as a stochastic kernel (regular
conditional probability) from  $\mathds{X}$ to $\mathds{Y}$, such that
$Q(\,\cdot\,|x)$ is a probability measure on the (Borel)
$\sigma$-algebra ${\cal B}(\mathds{Y})$ on $\mathds{Y}$ for every $x
\in \mathds{X}$, and $Q(A|\,\cdot\,): \mathds{X}\to [0,1]$ is a Borel
measurable function for every $A \in {\cal B}(\mathds{Y})$. Let a
decision maker (DM) be located at the output of an observation channel
$Q$, with inputs $X_t$ and outputs $Y_t$.  Let $\mathds{U}$, the action space, be a Borel
subset of some Euclidean space. An {\em admissible policy} $\gamma$ is a
sequence of control functions $\{\gamma_t,\, t\in \Zplus\}$ such
that $\gamma_t$ is measurable with respect to the $\sigma$-algebra
generated by the information variables
\[
I_t=\{Y_{[0,t]},U_{[0,t-1]}\}, \quad t \in \mathds{N}, \quad
  \quad I_0=\{Y_0\}.
\]
where
\begin{equation}
\label{eq_control}
U_t=\gamma_t(I_t),\quad t\in \Zplus
\end{equation}
are the $\mathds{U}$-valued control
actions and 
\[Y_{[0,t]} = \{Y_s,\, 0 \leq s \leq t \}, \quad U_{[0,t-1]} =
  \{U_s, \, 0 \leq s \leq t-1 \}.\]

\noindent We define $\Gamma$ to be the set of all such admissible policies.

The joint distribution of the state, control, and observation
processes is determined by (\ref{eq_control}) and the following
relationships:
\[  \Pr\bigl( (X_0,Y_0)\in B \bigr) =  \int_B P(dx_0)Q(dy_0|x_0), \quad B\in \mathcal{B}(\mathds{X}\times\mathds{Y}), \]
where $P$ is the (prior) distribution of the initial state $X_0$, and
\begin{eqnarray*}
\label{eq_evol}
 \Pr\biggl( (X_t,Y_t)\in B \, \bigg|\, (X,Y,U)_{[0,t-1]}=(x,y,u)_{[0,t-1]} \biggr)
 \\
 = \int_B \mathcal{T}(dx_t|x_{t-1}, u_{t-1})Q(dy_t|x_t),  B\in \mathcal{B}(\mathds{X}\times
\mathds{Y}), t\in \mathds{N},
\end{eqnarray*}
where $\mathcal{T}(\cdot|x,u)$ is a stochastic kernel from $\mathds{X}\times
\mathds{U}$ to $\mathds{X}$ and $Q(\cdot | x)$ is a stochastic kernel from $\mathds{X}$ to $\mathds{Y}$.

We let the objective of the agent be the minimization of the cost for the static or single stage case,
  \begin{align*}
    J(P,Q,\gamma)=E_P^{Q,\gamma}\left[c(X_0,U_0)\right]
  \end{align*}
 \noindent over the set of admissible policies $\gamma\in\Gamma$, where $c:\mathds{X}\times\mathds{U}\to\R$ is a Borel-measurable stage-wise cost function and $E_P^{Q,\gamma}$ denotes the expectation with initial state probability measure $P$ and measurement channel $Q$ under policy $\gamma$. Note that $P\in\mathcal{P}(\mathds{X})$, where we let $\mathcal{P}(\mathds{X})$ denote the set of probability measures on $\mathds{X}$.

For the multi-stage case, we will discuss the discounted cost infinite horizon setting, with the following cost criterion to be minimized.
  \begin{align*}
  &  J_{\beta}(P,Q,\gamma)= E_P^{Q,\gamma}\left[\sum_{t=0}^{\infty} \beta^t c(X_t,U_t),\right]
  \end{align*}
for some $\beta \in (0,1)$.

We define the optimal cost for the single-stage and the discounted infinite horizon as a function of the priors as
\begin{align*}
  J^*(P,Q)&=\inf_{\gamma\in\Gamma} J(P,Q,\gamma),\\
  J_{\beta}^*(P,Q)&=\inf_{\gamma\in\Gamma} J_{\beta}(P,Q,\gamma)
\end{align*}
respectively.

Note that for the discounted infinite horizon case, the cost function is also affected by the transition kernel $\mathcal{T}$. Thus, in the following, we may sometimes use $J_{\beta}(P,\mathcal{T},\gamma)$ and $J^*(P,\mathcal{T})$ instead of $J_{\beta}(P,Q,\gamma)$ and $J^*(P,Q)$ depending on the context.

The focus of the paper will be to address the following problems:

\noindent{\bf Problem P1: Continuity of $J^*(P,Q)$ and $J_\beta^*(P,Q)$  on the space of prior distributions.}
 Suppose ${\{P_n, n \in \mathds{N}\}}$ is a
sequence of priors converging in some sense to $P$. When does $P_n \to P$
imply $J^*(P_n,Q)\to J^*(P,Q)$ or $J_\beta^*(P_n,Q)\to J_\beta^*(P,Q)$?

\noindent{\bf Problem P2: Robustness to incorrect priors} \
 A problem of major practical importance is robustness of an optimal controller to modeling errors. Suppose that an optimal policy is constructed according to a model which is incorrect: how does the application of the control to the true model affect the system performance and does the error decrease to zero as the models become closer to each other? In particular, suppose that $\gamma_n$ is an optimal policy designed for $P_n$, an incorrect prior model for a true model $P$. Is it the case that if $P_n \to P$ then $J(P, Q, \gamma_n) \to  J^*(P,Q)$?

\noindent{\bf Problem P3:  Empirical consistency of optimal costs designed under learned priors.}
 Let $P\in\mathcal{P}(\mathds{X})$ be a fixed initial distribution for some state variable $X$ which is unknown to a decision maker. Suppose that the decision maker learns $\widetilde{P}_n$, its estimate of $P$, from the collection of empirical observations of realizations from i.i.d. random variables $X_1, X_2, \cdots $ distributed with $P$, and applies an optimal policy for the control problem with initial distribution $\widetilde{P}_n$ (i.e., through a {\it plug-in} or {\it separated} controller design). Defining for every (fixed) Borel $B \subset \mathbb{X}$, and $n \in
\mathbb{N}$, the empirical occupation measures
\[
\widetilde{P}_n(B)=
\frac{1}{n}\sum_{i=1}^{n} 1_{\{X_i \in B\}},
\]
do we have that the optimal cost computed for the empirical measures converges to the true optimal cost (for the case where the true $P$ is known) as $n \to \infty$, almost surely?


Here is a summary of the rest of the paper: In the following two subsections we present a literature review and some basic properties with regard to the convergences of probability measures. In Section \ref{ch:Prior}, we study the continuity properties of the optimal cost functions with respect to prior measures under different convergence notions for both the single-stage and multi-stage settings. In Section \ref{robustness}, we use the results from Section \ref{ch:Prior} to obtain robustness results on control policies designed under incorrect prior estimates. Finally, in Section \ref{empiricalLearning}, an application of the results to control systems where the prior measures are estimated through empirical measurements is presented. 

\subsection{Literature review}\label{ch:Background}
The $H_{\infty}$ criterion in robust control \cite{zames1981feedback} \cite{basbern} \cite{zhou1996robust} addresses the problem of robustness of control policies with respect to unmodeled dynamics. The goal in robust control is to design control policies that work sufficiently well for systems with model or disturbance uncertainty. 

Researchers have developed robust controllers through a game formulation, where the minimizer is the controller and the maximizer is the uncertainty, and have established the equivalence with a risk-sensitive cost minimization for a class of systems \cite{jacobson1973optimal} \cite{whittle1991risk} \cite{dupuis2000robust}. Through such a formulation, and by Legendre-type transforms, the operational use of the relative entropy methods have come to the literature; see e.g. \cite[Eqn. (4)]{dai1996connections} or \cite[Eqns. (2)-(3)]{dupuis2000robust}. Here, one selects a nominal system which satisfies a relative entropy bound between the actual measure and the nominal measure, solves a risk sensitive optimal control problem, and this solution provides an upper bound for the original system. As such, a common approach in robust stochastic control has been to model the stochastic disturbance affecting a system and consider perturbations which lead to finite deviations according to the Kullback-Leibler divergence (or relative entropy) between the actual measure and a reference measure, or embed the uncertainty with a penalty term imposed on the cost function under such a distance measure, see e.g. \cite{dupuis2000robust,CharalambousRezaei07, dupuis2000kernel}. Along a similar theme, \cite{lam2016} studies an optimization problem for the expected cost of an uncontrolled i.i.d. model under relative entropy bounds for the probability measures on the state variables; this can be considered to be similar to the setup considered in our paper where the convergence notion is in the relative entropy on the priors, as the considered process is i.i.d. We note here that the relative entropy is a very restrictive {\it distance} measure (note though that this does not define a metric) and in particular, through Pinsker's inequality \cite[Lemma 5.2.8]{GrayInfo}, it is stronger than even total variation which has also been studied in robust stochastic control: \cite{tzortzis2015dynamic} has studied a min-max formulation for robust control where the one-stage transition kernel belongs to a ball under the total variation metric for each state action pair and develops a dynamic programming based solution for both finite and discounted cost infinite horizon problems. Further related work with model uncertainty includes \cite{oksendal2014forward, benavoli2011robust}, with some further work in the economics literature \cite{hansen2001robust, gossner2008entropy}. 

The results are also related to the input estimation problem from finitely many samples, as reviewed in \cite{goeva2016reconstructing} and empirical risk minimization \cite{zhu2015risk}, although in our context, we will investigate robustness in the context of a {\it separated} design: An input model is estimated through empirical data, and an optimal policy is constructed with the assumption that the estimated model is correct (see e.g. \cite{KellyWeakConv} for an application). Can we guarantee empirical (asymptotic) consistency? We will discuss a number of general results in Section \ref{empiricalLearning} and observe that in many situations empirical consistency may not hold. 

On continuity properties in prior measures, \cite{WuVerdu} and \cite{WuVer11} have studied the special case of minimum mean-square estimation, that is with $c(x,u)=(x-u)^2$ for a single-stage problem across additive noisy channels of the form $y=x+w$, and established conditions leading to continuity or upper semi-continuity properties under weak convergence and Wasserstein metrics. 

Related work also includes the recent studies \cite{yuksel12:siam} and \cite{YukselBaker}; \cite{yuksel12:siam} considers various topologies on the sets of observation channels and quantizers in partially observed stochastic control and provides some supporting results, whereas \cite{YukselBaker} presents a number of continuity properties for single-stage stochastic control problems. 

The problems we consider are also related to, in the control-free context, the theory and applications of non-linear filtering with incorrect initial specifications. Here, the problem is to identify conditions on when an incorrectly initialized non-linear filter asymptotically gets corrected with the accumulation of additional measurements; these often require strong ergodicity properties of the Markov process \cite{budhiraja1997exponential, budhiraja1999exponential, chigansky2009intrinsic} or regularity properties (such as absolute continuity) of incorrect prior with respect to the true one and conditions on the measurement processes \cite{Ramon2008discrete}. 

\subsection{Convergence of probability measures and some supporting results}
Three important notions of convergences for sets of probability measures are weak convergence, setwise convergence and convergence under total variation (see, e.g., \cite{Hernandez} and \cite{yuksel12:siam}). For some $N\in\N$ a sequence $\{\mu_n,n\in\N\}$ in $\mathcal{P}(\R^N)$ is said to converge to $\mu\in\mathcal{P}(\R^N)$ \emph{weakly} if
  \begin{align}\label{converge}
    \int_{\R^N}c(x)\mu_n(dx) \to \int_{\R^N}c(x)\mu(dx)
  \end{align}

  \noindent for every continuous and bounded $c:\R^N \to \R$.
$\{\mu_n\}$ is said to converge \emph{setwise} to $\mu\in\mathcal{P}(\R^N)$ if (\ref{converge}) holds
\noindent for all measurable and bounded $c:\R^N \to \R$. Setwise convergence can also be equivalently defined through pointwise convergence on Borel subsets of $\R^N$, that is, $\mu_n(A)\to\mu(A) \text{ for all } A\in\mathcal{B}(\R^N)$.

  For probability measures $\mu,\nu\in\mathcal{P}(\R^N)$, the \emph{total variation} metric is given by
  \begin{align*}
    \|\mu-\nu\|_{TV}&=2\sup_{B\in\mathcal{B}(\R^N)}|\mu(B)-\nu(B)|\\
    &=\sup_{f:\|f\|_\infty \leq 1}\left|\int f(x)\mu(\dd x)-\int f(x)\nu(\dd x)\right|,
  \end{align*}

  \noindent where the supremum is taken over all measurable real $f$ such that \newline $\|f\|_\infty=\sup_{x\in\R^N}|f(x)|\leq 1$. A sequence $\{\mu_n\}$ is said to converge in total variation to $\mu\in\mathcal{P}(\R^N)$ if $\|\mu_n-\mu\|_{TV}\to 0$.

We next introduce the Wasserstein metric. The \emph{Wasserstein metric} of order 1 for two distributions $\mu,\nu\in\mathcal{P}(\mathds{X})$ is defined as
\begin{align*}
  W_1(\mu,\nu) =\inf_{\eta \in \mathcal{H}(\mu,\nu)} \int_{\mathds{X}\times\mathds{X}} \eta(\dd x,\dd y) |x-y|,
\end{align*}
where $\mathcal{H}(\mu,\nu)$ denotes the set of probability measures on $\mathds{X}\times\mathds{X}$ with first marginal $\mu$ and second marginal $\nu$.

A sequence $\{\mu_n\}$ is said to converge in $W_1$ to $\mu \in \mathcal{P}(\R^N)$ if $W_1(\mu_n,\mu)  \to 0$. For compact $\mathds{X}$, the Wasserstein distance of order $1$ metrizes the weak topology on the set of probability measures on $\mathds{X}$ (see \cite[Theorem 6.9]{villani2008optimal}). For non-compact $\mathds{X}$ convergence in the $W_1$ metric implies weak convergence (in particular this metric bounds from above the Bounded-Lipschitz metric \cite[p.109]{villani2008optimal}, which metrizes the weak convergence).

The following result shows the relation between convergence of prior measures and convergence of joint measures of channel and initial distribution given a fixed channel $Q \in \mathcal{Q}$.

The joint measure $PQ$ is induced on $(\mathds{X}\times\mathds{Y}, \mathcal{B}(\mathds{X}\times\mathds{Y}))$, for $Q \in \mathcal{Q}$ and $P \in \mathcal{P}(\mathds{X})$ where $\mathcal{Q}$ is a set of communication channels,
\begin{equation}
 PQ(A)=\int_A Q(\dd y|x)P(\dd x),~A\in\mathcal{B}(\mathds{X}\times\mathds{Y}).\nonumber
\end{equation}

The following is a result that will be used later in the paper. 

 \begin{lemma} \label{lemma:joint} 
 Let $Q:\mathcal{B}(\mathds{Y})\times\mathds{X}\to[0,1]$ be a stochastic kernel on $\mathds{Y}$ given $\mathds{X}$.
    \begin{enumerate}[label=(\roman*)]
      \item Assume that $Q(dy|x)$ is weakly continuous in $x$ in the sense that $\int Q(dy|x) c(y)$ is continuous in $x$ for every continuous and bounded $c$. If $\{P_n,n\in\N\}$ in $\mathcal{P}(\mathds{X})$ converges to $P\in\mathcal{P}(\mathds{X})$ weakly then $P_nQ\to PQ$ weakly.
       \item If $\{P_n,n\in\N\}$ in $\mathcal{P}(\mathds{X})$ converges to $P\in\mathcal{P}(\mathds{X})$ setwise then $P_nQ\to PQ$ setwise.
      \item If $\{P_n,n\in\N\}$ in $\mathcal{P}(\mathds{X})$ converges to $P\in\mathcal{P}(\mathds{X})$ in total variation then $P_nQ\to PQ$ in total variation. In particular, ${\|P_nQ-PQ\|_{TV}=\|P_n-P\|_{TV}}$.
    \end{enumerate}
  \end{lemma}
\begin{proof}
    \begin{enumerate}[label=(\roman*)]
      \item 
The sequence $\{f_n\}$ is said to converge continuously to $f$ when $\lim_{n \to \infty} f_n(x_n)=f(x)$ for any $x_n \to x$.

 First we show that under the given assumptions, and with any continuous and bounded $c:\mathds{X}\times\mathds{Y}\to\R$, $\int_\mathds{Y}c(x,y)Q(\dd y|x)$ is continuous in $x$.  Let $x_n\to x$ in $\mathds{X}$, then $c(x_n,y)\to c(x,y)$ for all $y\in\mathds{Y}$ by continuity of $c$. In particular, with
$c_n(y):=c(x_n,y)$ and $c(y):=c(x,y)$, it follows that as $y_n \to y$, $c_n(y_n) \to c(y)$; that is, $c_n(\cdot)$ {\it continuously converges} to $c(\cdot)$ as $(x_n,y_n) \to (x,y)$ by \cite[Theorem~3.3]{serfozo82} or  \cite[Theorem~3.5]{Langen81}, it follows that the continuity property
\[\lim_{n\to\infty}\int_\mathds{Y}c(x_n,y)Q(\dd y|x_n)=\int_\mathds{Y}c(x,y)Q(\dd y|x)\]
holds.
  
      Then, we have,
      \begin{align}
        & \lim_{n\to\infty}\int_{\mathds{X}\times\mathds{Y}}c(x,y)P_nQ(\dd x, \dd y) \nonumber \\
        &= \lim_{n\to\infty}\int_\mathds{X}\left(\int_\mathds{Y}c(x,y)Q(\dd y|x)\right)P_n(\dd x) \nonumber\\
        &=\int_\mathds{X}\left(\int_\mathds{Y}c(x,y)Q(\dd y|x)\right)P(\dd x) \nonumber \\
        &=\int_{\mathds{X}\times\mathds{Y}}c(x,y)PQ(\dd x, \dd y) \label{serfozoWeakC}
      \end{align}
      where we have applied Fubini's theorem and the fact that $\int_\mathds{Y}c(x,y)Q(\dd y|x)$ is a bounded and continuous function in $x$ under the given assumptions.


      \item Given any measurable and bounded $c:\mathds{X}\times\mathds{Y}\to\R$ we have
      \begin{align*}
        \lim_{n\to\infty}\int_{\mathds{X}\times\mathds{Y}}c(x,y)P_nQ(\dd x, \dd y)&= \lim_{n\to\infty}\int_\mathds{X}\left(\int_\mathds{Y}c(x,y)Q(\dd y|x)\right)P_n(\dd x)\\
        &=\int_\mathds{X}\left(\int_\mathds{Y}c(x,y)Q(\dd y|x)\right)P(\dd x)\\
        &=\int_{\mathds{X}\times\mathds{Y}}c(x,y)PQ(\dd x, \dd y)
      \end{align*}
      where we have applied Fubini's theorem and the fact that $\int_\mathds{Y}c(x,y)Q(\dd y|x)$ is a bounded and measurable function of $x$ under the given assumptions (see Proposition 7.29 in \cite{bertsekas78}).
       \item {Observe the following,
      \begin{align*}
&        \|P_nQ-PQ\|_{TV} \\
        &=\sup_{f:\|f\|_\infty \leq 1}\left|\int_{\mathds{X}\times\mathds{Y}} f(x,y)P_nQ(\dd x, \dd y)-\int_{\mathds{X}\times\mathds{Y}} f(x,y)PQ(\dd x, \dd y)\right|, \\
        &=\sup_{f:\|f\|_\infty \leq 1}\left|\int_{\mathds{X}}\left(\int_\mathds{Y} f(x,y)Q(\dd y | x)\right)P_n(\dd x) - \int_{\mathds{X}}\left(\int_\mathds{Y} f(x,y)Q(\dd y | x)\right)P(\dd x)\right| \\
        &\leq \sup_{\bar{f}:\|\bar{f}\|_\infty \leq 1}\left|\int_{\mathds{X}}\bar{f}(x) P_n(\dd x) - \int_{\mathds{X}}\bar{f}(x)P(\dd x)\right| \\     
        &=  \|P_n-P\|_{TV}
      \end{align*}
      where we used the fact that $\left|\bar{f}(x)\right| := \left|\int_\mathds{Y} f(x,y)Q(\dd y | x)\right| \leq 1$ for every $x \in \mathds{X}$. \\}      

      In addition, as $P(B)=PQ(B\times\mathds{Y})$ for any Borel $B$ and any $P\in\mathcal{P}(\mathds{X})$, then we have,
      \begin{align*}
        \|P_n-P\|_{TV}&=2\sup_{B\in\mathcal{B}(\mathds{X})}\left|P_n(B)-P(B)\right|\\
        &=2\sup_{B\in\mathcal{B}(\mathds{X})}\left|P_nQ(B\times\mathds{Y})-PQ(B\times\mathds{Y})\right|\\
        &\leq 2\sup_{A\in\mathcal{B}(\mathds{X}\times\mathds{Y})}\left|P_nQ(A)-PQ(A)\right|\\
        & =\|P_nQ-PQ\|_{TV}.
      \end{align*}
      Therefore ${\|P_nQ-PQ\|_{TV}=\|P_n-P\|_{TV}}\to 0.$ 
         \end{enumerate}
  \end{proof}

\section{Continuity of Optimal Cost with respect to Convergence of Prior Probability Measures}\label{ch:Prior}
In this section, we study the continuity in the priors.

\subsection{Single-stage setup}
 First, we consider a single-stage setup. The initial measure will be denoted by  $P\in\mathcal{P}(\mathds{X})$, the measurement channel will be $Q\in\mathcal{Q}$ and the control action will be from the admissible control policies, $\Gamma$. Here, $\Gamma$ is just the set of functions from $\mathds{Y}$ to $\mathds{U}$ that are measurable with respect to the $\sigma$-algebra generated by $I_0=\{Y_0\}$.

The optimal single stage cost will be defined as,
 \begin{align*}
    J^*(P,Q)&=\inf_{\gamma\in\Gamma} \int_{\mathds{X}\times\mathds{Y}} c(x,\gamma(y))Q(\dd y|x)P(\dd x).
  \end{align*}

  The following is a useful result, to be used throughout the paper. 
  
  \begin{theorem}\cite[Theorem 3.1]{yuksel12:siam}
    Let $c$ be continuous in $u$ for every $x$, and $\mathds{U}$ are compact. Then, for the static problem, there exists an optimal control policy for any channel $Q$.
  \end{theorem}
For the rest of the paper, we will use optimal policies whenever they exist otherwise we will use $\epsilon$-optimal policies.

\subsubsection{Weak convergence}\label{section:weakcont}
We will first observe that $J^*(P,Q)$ is not always continuous in initial measures under weak convergence, then we will present an upper semi-continuity result and some sufficiency continuity conditions for weak convergence of prior measures.

\begin{theorem}\label{weak_count}
    Let a channel $Q\in\mathcal{Q}$ be given. $J^*(P,Q)$ is not necessarily continuous in $P$ under weak convergence. This holds even when $\mathds{X}$, $\mathds{Y}$, and $\mathds{U}$ are compact and when $c$ is bounded and continuous on $\mathds{X}\times \mathds{U}$.
  \end{theorem}
\begin{proof}
We prove the result with a counter example. Let $\mathds{X}=\mathds{U}=\mathds{Y}=[0,1]$ and $c(x,u)=(x-u)^2$. The optimal policy for this cost function is $\gamma^*(y)=E[X|Y=y]$. The measurement channel is given in the following form,
\[Q(\cdot|x)=\frac{1}{2}\delta_x(\cdot)+\frac{1}{2}\delta_0(\cdot).\]
Let the prior distributions $P$ and $P_n$ are given by
\begin{align*}
&P=\frac{1}{2}\delta_0+\frac{1}{2}\delta_1\\
&P_n=\frac{1}{2}\delta_{\frac{1}{n}}+\frac{1}{2}\delta_1
\end{align*}
Note that $P_n \to P$ weakly as $n \to \infty$. Now, we calculate the optimal control actions for prior model $P$.
\begin{align*}
&\gamma^*(0)=E[X|Y=0]=Pr(X=1|Y=0)=Pr(Y=0|X=1)\frac{Pr(X=1)}{Pr(Y=0)}=\frac{1}{3}\\
&\gamma^*(1)=E[X|Y=1]=Pr(X=1|Y=1)=Pr(Y=1|X=1)\frac{Pr(X=1)}{Pr(Y=1)}=1.
\end{align*} 
Thus, the optimal cost for $P$ can be calculated as
\begin{align*}
J^*(P,Q)&=E[(X-\gamma^*(Y))^2]\\
&=Pr(X=0)Pr(Y=0|X=0)(0-\frac{1}{3})^2+Pr(X=1)Pr(Y=0|X=1)(1-\frac{1}{3})^2\\
&\quad+Pr(X=0)Pr(Y=1|X=0)(0-1)^2+Pr(X=1)Pr(Y=1|X=1)(1-1)^2\\
&=\frac{1}{18}+\frac{1}{9}=\frac{1}{6}
\end{align*}
The optimal control actions for $P_n$ can be calculated as
\begin{align*}
\gamma_n^*(0)&=E_n[X|Y=0]=\frac{1}{n}Pr(X=1/n|Y=0)+1Pr(X=1|Y=0)\\
&=\frac{1}{n}\frac{Pr(Y=0|X=1/n)Pr(X=1/n)}{Pr(Y=0)}+\frac{Pr(Y=0|X=1)Pr(X=1)}{Pr(Y=0)}\\
&=\frac{n+1}{2n}\\
\gamma_n^*(1)&=E_n[X|Y=1]=\frac{1}{n}Pr(X=1/n|Y=1)+1Pr(X=1|Y=1)\\
&=\frac{1}{n}\frac{Pr(Y=1|X=1/n)Pr(X=1/n)}{Pr(Y=1)}+\frac{Pr(Y=1|X=1)Pr(X=1)}{Pr(Y=1)}\\
&=1\\
\gamma_n^*(1/n)&=E_n[X|Y=1/n]=\frac{1}{n}Pr(X=1/n|Y=1/n)+1Pr(X=1|Y=1/n)\\
&=\frac{1}{n}\frac{Pr(Y=1/n|X=1/n)Pr(X=1/n)}{Pr(Y=1/n)}+\frac{Pr(Y=1/n|X=1)Pr(X=1)}{Pr(Y=1/n)}\\
&=\frac{1}{n}.
\end{align*}
Using the optimal control actions, we can calculate the optimal cost for $P_n$.
\begin{align*}
J^*(P_n,Q)&=E_n[(X-\gamma^*(Y))^2]\\
&=Pr(X=0)Pr(Y=0|X=0)(0-\frac{n+1}{2n})^2 \\
& \quad \quad \quad \quad +Pr(X=1)Pr(Y=0|X=1)(1-\frac{n-1}{2n})^2\\
&=\frac{1}{2}(\frac{n+1}{2n})^2+\frac{1}{4}(\frac{n-1}{2n})^2=\frac{3n^2+2n+3}{16n^2}.
\end{align*} 
We can see that as $n \to \infty$, $J^*(P_n,Q)\to \frac{3}{16} \neq \frac{1}{6}=J^*(P,Q)$.
\end{proof}

\adk{The channel model in the example is a channel which either sends full information across the channel without error, or it provides no information and this is an infinite-capacity erasure channel: This channel has a practical significance as in practice we may have package drops during the transmission of the state variable which causes controller not to get any data from the system at random times. 

 We now provide another example which may have further practical significance. Quantizer channels are used often in practice (see \cite[Section 5]{yuksel12:siam}). In the following example, we show that continuity may not hold for quantizer channels either.

Before the example we define quantizers: An $M$-cell vector quantizer, $q$, is a (Borel) measurable mapping
from $\mathds{X}=\mathds{R}^n$ to the finite set $\{1, 2,...,M\}$, characterized by a measurable partition
$\{B1, B2,...,BM \}$ such that $B_i = \{x : q(x) = i\}$ for  $i = 1,\dots,M$. The $B_i$ are called
the cells (or bins) of $q$

A quantizer $q$ with cells $\{B_1,\dots,B_M \}$, however, can also be characterized as a
stochastic kernel $Q$ from $\mathds{X}$ to $\{1,...,M\}$ defined by
\begin{align*}
Q(i|x)=1_{x\in B_i},\qquad i = 1,\dots,M,
\end{align*}
so that $q(x) = \sum_{i=1}^M Q(i|x)$.

\begin{exmp}\label{weak_count2}
Assume the prior distributions are given by
\begin{align*}
&P(\cdot)=\frac{1}{2}\delta_{\frac{1}{2}}(\cdot)+\frac{1}{2}\delta_{1}(\cdot)\\
&P_n(\cdot)=\frac{1}{2}\delta_{\frac{1}{2}-\frac{1}{n}}(\cdot)+\frac{1}{2}\delta_{1}(\cdot)
\end{align*} 
and the channel is a quantizer with $B_1=[0,\frac{1}{2})$ and $B_2= [\frac{1}{2},1]$ where the range of the quantizer is $\{0,1\}$ i.e.
\begin{align*}
Q(0|x)=1_{x\in[0,\frac{1}{2})},\qquad Q(1|x)=1_{x\in[\frac{1}{2},1]}.
\end{align*}
If the cost function is given by $c(x,u)=(x-u)^2$ then the optimal controls are
\begin{align*}
&\gamma^*(0)=0\qquad \gamma^*(1)=3/4\\
&\gamma_n^*(0)=\frac{1}{2}-\frac{1}{n}\qquad \gamma_n^*(1)=1
\end{align*}
under this setup the value functions can be calculated as
\begin{align*}
&J^*(P_n,Q)=0\\
&J^*(P,Q)=P(X=1/2)(1/2-3/4)^2+P(X=1)(1-3/4)^2=1/16
\end{align*}
which shows that the optimal cost is not continuous under the weak convergence of priors when the channel is a quantization channel.
\hfill $\diamond$
\end{exmp}

}

Now we show that the optimal cost is upper semi-continuous under weak convergence of priors. The next lemma, building on \cite{yuksel12:siam}, shows that the optimal cost is unchanged when $\gamma$ is restricted to the class of continuous policies. A brief proof is presented in the appendix. 

  \begin{lemma}\label{lemma:lusin}
    Let $\mu$ be an arbitrary probability measure on $(\mathds{X}\times\mathds{Y},\mathcal{B}(\mathds{X}\times\mathds{Y}))$ and let $\mathcal{C}$ be the set of continuous functions from $\mathds{Y}$ to $\mathds{U}$. If $\mathds{U}$ is convex and $c(x,u)$ is non-negative, measurable and bounded in $\mathds{X}\times\mathds{U}$ then

    \begin{align*}
      \inf_{\gamma\in\Gamma} \int c(x,\gamma(y))\mu(\dd x, \dd y)=\inf_{\gamma\in\mathcal{C}} \int c(x,\gamma(y))\mu(\dd x, \dd y)
    \end{align*}
  \end{lemma}

  We now show that the optimal cost $J^*(P,Q)$ is upper semi-continuous under weak convergence on the space of initial distributions. The following theorem is related to Theorem 3 of Wu and Verd\'u \cite{WuVerdu}.

  \begin{theorem}
    Let a channel $Q\in\mathcal{Q}$ be given. If $Q(dy|x)$ is weakly continuous in $x$ in the sense that $\int Q(dy|x) c(y)$ is continuous in $x$ for every continuous and bounded $c$, and if $\mathds{U}$ is convex and $c(x,u)$ is non-negative, continuous and bounded in $\mathds{X}\times\mathds{U}$ then $J^*(P,Q)$ is upper semi-continuous on $\mathcal{P}(\mathds{X})$ under weak convergence.
  \end{theorem}

  \begin{proof}
    Recall from the statement of Lemma~\ref{lemma:lusin} that $\mathcal{C}$ denotes the set of all continuous functions from $\mathds{Y}$ to $\mathds{U}$. Let $P_n\to P$ in $\mathcal{P}(\mathds{X})$ weakly. Then 
    \begin{align*}
      \limsup_{n\to\infty}\inf_{\gamma\in\Gamma} \int_{\mathds{X}\times\mathds{Y}} & c(x,\gamma(y))P_nQ(\dd x,\dd y)\\
      &=\limsup_{n\to\infty}\inf_{\gamma\in\mathcal{C}} \int_{\mathds{X}\times\mathds{Y}} c(x,\gamma(y))P_nQ(\dd x,\dd y)\\
      &\leq\inf_{\gamma\in\mathcal{C}} \limsup_{n\to\infty} \int_{\mathds{X}\times\mathds{Y}} c(x,\gamma(y))P_nQ(\dd x,\dd y)\\
      &=\inf_{\gamma\in\mathcal{C}} \int_{\mathds{X}\times\mathds{Y}} c(x,\gamma(y))PQ(\dd x,\dd y)\\
      &=\inf_{\gamma\in\Gamma} \int_{\mathds{X}\times\mathds{Y}} c(x,\gamma(y))PQ(\dd x,\dd y),
    \end{align*}

    \noindent where the first and last equality rely on Lemma~\ref{lemma:lusin} and the second-to-last equality holds as $c(x,\gamma(y))$ is bounded and continuous for $\gamma\in\mathcal{C}$.
  \end{proof}

The following result shows us that if we put some continuity restrictions on the measurement channel then we can guarantee the continuity of optimal single stage cost function under weak convergence of prior measures.
 \begin{assumption}\label{assumptionTV}
$Q$ is continuous in total variation in the sense that as $x_n \to x$, $\|Q(dy|x_n) - Q(dy|x) \|_{TV} \to 0$.
\end{assumption}

  \begin{theorem}\label{thm:weaksuff}
  Let a cost function $c:\mathds{X}\times\mathds{U}\to \R$ be given and let Assumption \ref{assumptionTV} holds, $c(x,u)$ be bounded and continuous on $\mathds{X}\times\mathds{U}$ and $\mathds{U}$ be compact.  If $P_n\to P$ weakly then $J^*(P_n,Q)\to J^*(P,Q)$.
 \end{theorem}   

  \begin{proof}
  We first show that under the stated assumptions, the following holds.
    \begin{align}
      \lim_{k\to\infty} \sup_\gamma \left| \int Q(\dd y| x_k)c(x_k,\gamma(y)) - \int  Q(\dd y| x)c(x,\gamma(y)) \right| =0 \label{equiC}
    \end{align}
    for $x_k\to x$. To see this, write
        \begin{eqnarray}
&&      \lim_{k \to \infty} \sup_\gamma \left| \int Q(\dd y| x_k)c(x_k,\gamma(y)) - \int  Q(\dd y| x)c(x,\gamma(y)) \right| \nonumber \\
&&   \leq \lim_{k \to \infty}\sup_\gamma \left| \int Q(\dd y| x_k)c(x_k,\gamma(y)) - \int  Q(\dd y| x)c(x_k,\gamma(y)) \right| \label{firstT} \\
&&    + \lim_{k \to \infty}\sup_\gamma \left|  \int  Q(\dd y| x)c(x_k,\gamma(y))  - \int  Q(\dd y| x)c(x,\gamma(y)) \right| \label{secondT}
    \end{eqnarray}
The term in (\ref{firstT}) converges to $0$ by Assumption (\ref{assumptionTV}). Since we have,
\begin{align*}
 &\lim_{k \to \infty}\sup_\gamma \left| \int Q(\dd y| x_k)c(x_k,\gamma(y)) - \int  Q(\dd y| x)c(x_k,\gamma(y)) \right| \\
&\quad \leq  \lim_{k\to\infty} \|Q(.|x_k)-Q(.|x) \|_{TV}\\
&\quad=0
\end{align*}
By assumption the action space $\mathds{U}$ is compact. Therefore, $c(x_k,u) \to c(x,u)$ uniformly on  $\mathds{U}$ . Thus, (\ref{secondT}) also goes to 0.

 Let the optimal policies be given by $\gamma_n^*$ and $\gamma^*$ for initial distributions $P_n$ and $P$ respectively. We now consider two cases: \newline
    Firstly, if $J^*(P_n,Q)< J^*(P,Q)$ then
    \begin{align}
      \label{eq:ineq1}
      \begin{split}
        J^*(P,Q)-J^*(P_n,Q)\leq J(P,Q,\gamma_n^*)-J(P_n,Q,\gamma_n^*).
      \end{split}
    \end{align}
    Secondly, if $J^*(P,Q)<J^*(P_n,Q)$ then
    \begin{align}
      \label{eq:ineq2}
      \begin{split}
        J^*(P_n, Q)-J^*(P,Q)\leq J(P_n,Q,\gamma^*)-J(P,Q,\gamma^*).
      \end{split}
    \end{align}
    Inequalities \eqref{eq:ineq1} and \eqref{eq:ineq2} are combined to give
    \begin{align}\label{eq:max}
      \begin{split}
        |J^*(P,Q&)-J^*(P_n,Q)|\\
        &\leq\max(J(P, Q, \gamma_n^*)-J(P_n,Q,\gamma_n^*),J(P_n, Q, \gamma^*)-J(P,Q,\gamma^*)).
      \end{split}
    \end{align}
    
   Observe the following:
    \begin{align}
      &|J^*(P_n,Q)-J^*(P,Q)|\nonumber\\
      & \quad\leq  \max\bigg(\left| \int (P_n-P)(\dd x) \int Q(\dd y| x)c(x,\gamma_n^*(y))\right|, \nonumber \\ 
& \quad  \quad\quad  \quad  \quad  \quad  \quad       \left|\int  (P_n-P)(\dd x) \int Q(\dd y| x)c(x,\gamma^*(y)) \right| \bigg), \label{maxterm}
    \end{align}
    For continuity, we need \eqref{maxterm} to tend to $0$ as $n\to\infty$.

Let $F$ be a family of functions from normed linear spaces $\mathds{S}$ to $\mathds{T}$. The family $F$ is said to be {\it equicontinuous at a point} $x_0 \in S$ if, for every
$\epsilon > 0$, there exists a $\delta > 0$ such that $|f(x) - f(x_0)| \leq \epsilon$ for all $f\in F$ and for every $x$ such that $|x-x_0| \leq \delta$.
The family $F$ is said to be {\it equicontinuous} if it is equicontinuous at each $x\in S$. 
A consequence of weak convergence is that (see, e.g. Lemma C.1 in \cite{gupta2014existence}) if $f_n$ is an equicontinuous family, then, $P_n \to P$ weakly, implies that $\int (P_n - P)(dx) f_n(x) \to 0$. Condition (\ref{equiC}) ensures that the sequence of functions $\int Q(\dd y| x_k)c(x_k,\gamma_n(y))$ is equicontinuous as it gives us a uniform continuity over family of all admissible policies.  and the result follows.
  \end{proof}

\begin{exmp}\label{ex:additive}
  Consider the following additive noisy channel:
  \begin{align*}
    y=x+w,
  \end{align*}
  where $w \sim \mu$ with $\mu$ admitting a density, $\eta$, which is continuous. An example is the Gaussian density. Suppose that $\mathds{U}$ is compact and $c(x,u)$ is continuous and bounded in $\mathds{X}\times\mathds{U}$. For $x_k\to x$ we have that
  \begin{align}
    &\lim_{k\to\infty} \sup_\gamma \left| \int \eta(y- x_k)c(x_k,\gamma(y))\dd y - \int  \eta(y-x)c(x,\gamma(y)) \dd y \right|\nonumber\\
    &\quad \leq \lim_{k\to\infty} \sup_\gamma \left| \int (\eta(y- x_k)-\eta(y-x))c(x_k,\gamma(y))\dd y \right. + \dots \nonumber \\
    &\quad \quad +\left. \int  \eta(y-x)\left(c(x_k,\gamma(y))-c(x,\gamma(y))\right)\dd y \right| \nonumber\\
    &\quad \leq \|\eta(\cdot-x_k)-\eta(\cdot-x)\|_{TV}\cdot \|c\|_{\infty} + \dots \label{eq:addexamp1}\\
    &\quad \quad +\lim_{k\to\infty} \sup_\gamma \left| \int  \eta(y-x)\left(c(x_k,\gamma(y))-c(x,\gamma(y))\right)\dd y \right| \label{eq:addexamp2}\\
    &\quad =0 \nonumber,
  \end{align}
  where $\|c\|_\infty$ denotes the supremum norm of $c$. We note that the term in \eqref{eq:addexamp1} tends to zero since $\eta(\cdot-x_k)$ converges to $\eta(\cdot-x)$ pointwise and therefore by Scheff\'e's theorem it converges in $L_1$ and thus in total variation. Additionally, the term  in \eqref{eq:addexamp2} tends to zero since $c$ is uniformly continuous by assumptions. Therefore   Assumption~\ref{assumptionTV} is satisfied and Theorem~\ref{thm:weaksuff} holds. Thus, for a special but practically important class of channels weak convergence of priors is sufficient for continuity if further $c(x,u)$ is bounded and continuous on $\mathds{X}\times\mathds{U}$ and $\mathds{U}$ is compact.
  \hfill $\diamond$
\end{exmp}

We note here that a related result due to Wu and Verd\'u \cite{WuVerdu} establishes continuity of the MMSE error (that is with $c(x,u) = \|x-u\|^2$) under weak convergence when the channel is additive, the additive noise has a finite variance and it admits a continuous and bounded density function. In general, however, the following example shows that the boundedness condition cannot be relaxed even when the channel is non-informative, which can be viewed as an extreme form of regularity.

  \begin{exmp}\label{ex:weakcounter}
    Let $\mathds{X}=\mathds{U}=\R$, $\mathds{Y}=[0,1]$, and $c(x,u)=(x-u)^2$. With the given cost function this is a mean-square error problem; therefore, the optimal policy is $\gamma^*(y)=E[x|y]$. We let the channel be distributed uniformly on $[0,1]$, that is, $Q\sim U([0,1])$. Note that this channel is non-informative. Let $P_n$ be the following discrete distribution,
    \begin{align*}
      P_n=\left(\frac{1}{2}-\frac{1}{n}\right)\cdot\delta_{\frac{1}{n}}+\left(\frac{1}{2}-\frac{1}{n}\right)\cdot\delta_{-\frac{1}{n}}+\frac{1}{2n}\cdot\delta_{a_n}+\frac{1}{2n}\cdot\delta_{-a_n}
    \end{align*}
    where $\delta_s$ is the delta measure at point $s$, that is,
    \begin{align*}
        \delta_s(A)=1_{\{s\in A\}}
    \end{align*}
    for any Borel set $A$, and $a_n$ is the sequence of numbers in $\N$ defined by
    \begin{align*}
      a_n=\sqrt{n-\left(\frac{1}{n}+\frac{2}{n^2}\right)}.
    \end{align*}
    Clearly $P_n\to \delta_0$ weakly as for any bounded and continuous function $f$ we have
    \begin{align*}
      \int_\R f(x) &P_n(\dd x)=\\
      &\left(\frac{1}{2}-\frac{1}{n}\right)\cdot f\left(\frac{1}{n}\right)+\left(\frac{1}{2}-\frac{1}{n}\right)\cdot f\left(-\frac{1}{n}\right)+\frac{1}{2n}\cdot f(a_n)+\frac{1}{2n}\cdot f(-a_n)\\
      &\to f(0) = \int_\R f(x) \delta_0(\dd x)
    \end{align*}
    by boundedness and continuity. By symmetry and the non-informative nature of $Q$, the optimal policy is $\gamma^*(y)=E[X|Y]=0$ for all $P_n$ and for $P=\delta_0$. With initial distribution $P$, we have 
    \begin{align*}
      J^*(P,Q)&=E_P^{Q,\gamma^*}[(X-U)^2]\\
      &=E_P^{Q,\gamma^*}[(X-\gamma^*(Y))^2]\\
      &=E_P^{Q,\gamma^*}[(X)^2]=0.
    \end{align*}
    Whereas for all $n\in\N$ we have
    \begin{align*}
      J^*(P_n,Q)&=E_{P_n}^{Q,\gamma^*}[(X)^2]\\
      &=\left(\frac{1}{2}-\frac{1}{n}\right)\cdot\frac{1}{n^2}+\left(\frac{1}{2}-\frac{1}{n}\right)\cdot\frac{1}{n^2}+\frac{1}{2n}\cdot a_n^2+\frac{1}{2n}\cdot a_n^2\\
      &=\left(1-\frac{2}{n}\right)\cdot\frac{1}{n^2}+\frac{1}{n}\cdot\left(n-\left(\frac{1}{n}+\frac{2}{n^2}\right)\right)\\
      &=1,
    \end{align*}
    so $J^*(P_n,Q)\not\to J^*(P,Q)$ as $n\to\infty$. \hfill $\diamond$
  \end{exmp}
  
Now we present a result for continuity under the Wasserstein metric.
\begin{assumption}\label{abscont}
There exists a measurable non-negative function $f$ so that for some probability measure $P_Q$, the following (absolute continuity condition) holds:
\[Q(Y \in A | x) = \int_A f(x,y) P_Q(dy)\]
Furthermore, $\tilde{c}(x,y,u) := c(x,u) f(x,y)$ is such that
\[ |\tilde{c}(x',y,u) - \tilde{c}(x,y,u)| \leq \alpha |x' - x|\]
for all $y \in \mathds{Y}, u \in \mathds{U}$ and for some $\alpha \in \mathbb{R}_+$.
\end{assumption}

\begin{theorem}\label{absContThmStatic}
 Under Assumption \ref{abscont}, 
 \[ |J^*(P_n,Q) - J^*(P,Q)| \leq \alpha W_1(P_n,P),\] 
and thus, as $W_1(P_n,P) \to 0$, $J^*(P_n,Q) \to J^*(P,Q)$.
\end{theorem}

\begin{proof}

We first use the bound in (\ref{eq:max}) such that\begin{align*}
      \begin{split}
        |J^*(P,Q&)-J^*(P_n,Q)|\\
        &\leq\max(J(P, Q, \gamma_n^*)-J(P_n,Q,\gamma_n^*),J(P_n, Q, \gamma^*)-J(P,Q,\gamma^*)).
      \end{split}
    \end{align*}
    
Let $\mathds{P}_n$ denote a product measure on the space $\mathds{X}\times \mathds{X}$ such that its first marginal is $P$ and the second marginal is $P_n$. 
Then, for any $\gamma \in \Gamma$ we have
 \begin{align*}
& \int \mathds{P}_n(\dd x,\dd x') \bigg(\int Q(dy|x)c(x,\gamma(y)) - \int Q(dy|x')c(x',\gamma(y))\bigg) \\
  & =\int \mathds{P}_n(\dd x,\dd x') \left( \int P_Q(\dd y)   f(x, y) c(x,\gamma(y)) - \int P_Q(\dd y) f(x', y) c(x',\gamma(y)) \right)  \\
   & = \int \mathds{P}_n(\dd x,\dd x') \int P_Q(\dd y)  \left( f(x, y) c(x,\gamma(y)) - f(x', y) c(x',\gamma(y)) \right)  \\
& \leq  \int \mathds{P}_n(\dd x,\dd x') \int P_Q(\dd y) \alpha |x' - x| \\
&   \leq \int \mathds{P}_n(\dd x,\dd x') \alpha |x' - x|.
 \end{align*}
Optimizing over all such couplings $\mathds{P}_n$ completes the proof.
\end{proof}

\subsubsection{Setwise convergence}\label{section:setwisecounter}
$\quad$
\begin{theorem}\label{ex:setwise}
    Let a channel $Q\in\mathcal{Q}$ be given. $J^*(P,Q)$ is not necessarily continuous in $P$ under setwise convergence. This holds even when $\mathds{X}$, $\mathds{Y}$, and $\mathds{U}$ are compact and when $c$ is bounded and continuous in both $x$ and $u$.
  \end{theorem}
  \begin{proof}
    We present the following counterexample, building on \cite{yuksel12:siam}: Let $\mathds{X}=\mathds{Y}=\mathds{U}=[0,1]$ and let $c(x,u)=(x-u)^2$. For $n\in\N$ and $k=1,\dots,n$, we define
    \begin{align*}
      L_{n,k}=\left[\frac{2k-2}{2n},\frac{2k-1}{2n}\right),~R_{n,k}=\left[\frac{2k-1}{2n},\frac{k}{n}\right).
    \end{align*}
    For ease of notation, we shall let $L=\left\{y\in\cup_{k=1}^n L_{n,k}\right\}$ and $R=\left\{y\in\cup_{k=1}^n R_{n,k}\right\}$. Next, we define the square-wave function by 
    \begin{align*}
      h_n(t)=1_{\{t\in L\}}-1_{\{t\in R\}}.
    \end{align*}
    As $\int_0^1h_n(t) \dd t=0$ and $|h_n(t)|\leq 1$, the function
    \begin{align*}
      f_n(t)=(1+h_n(t))1_{\{t\in [0,1]\}}
    \end{align*}
    is a probability density function. 

    By the proof of the Riemann-Lebesgue lemma (for example, see Theorem 12.21 in \cite{wheeden77}), we have
    \begin{align*}
      \lim_{n\to\infty}\int_0^1 h_n(t)g(t) \dd t =0 \text{ for all } g\in L_1\left([0,1],\R\right),
    \end{align*}
    therefore
    \begin{align*}
      \lim_{n\to\infty}\int_0^1 f_n(t)g(t) \dd t =\int_0^1 g(t) \text{ for all } g\in L_1\left([0,1],\R\right).
    \end{align*}
    So if we let $P_n\sim f_n$ for $n\in\N$, we have that $P_n\to P\sim U([0,1])$ setwise. Next we let the channel be
    \begin{align*}
      Q(\cdot|x)\sim\frac{1}{2}\cdot\delta_x+\frac{1}{2}\cdot U([0,1]).
    \end{align*}
    For initial distribution $P$, the optimal policy is 
    \begin{align*}
      \gamma_P^*(y)=E[X|Y]=\frac{1}{2}\left(\frac{1}{2}+y\right).
    \end{align*}
    This gives
    \begin{align*}
      J^*(P,Q)=\frac{1}{16}.
    \end{align*}
    By tedious calculations (see Section \ref{calculationsSetwise}), the optimal policy for initial distribution $P_n$ is
    \begin{align*}
      \gamma_{P_n}^*(y)=
      \begin{cases}
        \frac{1}{2}-\frac{1}{4n} & \text{if } y\in\cup_{k=1}^n R_{n,k}\\
        \frac{1}{3}\cdot\left(\frac{1}{2}-\frac{1}{4n}\right)+\frac{2}{3}y & \text{if } y\in\cup_{k=1}^n L_{n,k}
      \end{cases}.
    \end{align*}
    This gives
    \begin{align*}
      J^*(P_n,Q)=\frac{1}{18}-\frac{1}{24n^2}.
    \end{align*}
    So we have $J^*(P_n,Q)\to\frac{1}{18}\neq\frac{1}{16}$ as $n\to\infty$, and we see that the optimal cost is clearly not continuous on the space of initial distributions under setwise convergence.
  \end{proof}

The next result shows that the optimal cost is upper semi-continuous under setwise convergence too. 

 \begin{theorem}
    Let a channel $Q$ be given. If $c(x,u)$ is non-negative, measurable and bounded in $\mathds{X}\times\mathds{U}$, then $J^*(P,Q)$ is upper semi-continuous on $\mathcal{P}(\mathds{X})$ under setwise convergence.
  \end{theorem}

  \begin{proof}
    Given a fixed channel $Q\in\mathcal{Q}$, let $P_n\to P$ in $\mathcal{P}(\mathds{X})$ setwise. Then 
    \begin{align*}
      \limsup_{n\to\infty}\inf_{\gamma\in\Gamma} \int_{\mathds{X}\times\mathds{Y}} & c(x,\gamma(y))Q(\dd y|x)P_n(\dd x)\\ &\leq \inf_{\gamma\in\Gamma} \limsup_{n\to\infty} \int_{\mathds{X}\times\mathds{Y}} c(x,\gamma(y))Q(\dd y|x)P_n(\dd x)\\
      &= \inf_{\gamma\in\Gamma} \int_{\mathds{X}\times\mathds{Y}} c(x,\gamma(y))Q(\dd y|x)P(\dd x),
    \end{align*}
where for the last equality we used Lemma \ref{lemma:joint} and the fact that $c$ is bounded and measurable.
  \end{proof}

\subsubsection{Continuity under total variation}
The proof of the following result builds on \cite[Theorem 3.4]{yuksel12:siam}.
  \begin{theorem} 
    The optimal cost $J^*(P,Q)$ is continuous on the set of input distributions, $\mathcal{P}(\mathds{X})$, under the topology of total variation. In other words, if $\|P_n-P\|_{TV}\to 0$, then $|J^*(P_n,Q)-J^*(P,Q)| \to 0$.
  \end{theorem}

 \begin{proof}
    Let $P_n\to P$ in total variation with a fixed channel $Q$. Recall from Section~\ref{section:intro} that $J(P',Q,\gamma')=E_{P'}^{Q,\gamma'}[c(x,u)]$, that is, the expected cost with initial distribution $P'\in\mathcal{P}(\mathds{X})$ and control policy $\gamma'\in\Gamma$. Let the optimal ($\varepsilon$-optimal) policies are given by $\gamma_n^*$ and $\gamma^*$ for initial distributions $P_n$ and $P$ respectively. Using the bound in \ref{eq:max} we write
\begin{align*}
|J^*(P_n,Q)-&J^*(P,Q)|\\
&\leq\max(J(P, Q, \gamma_n^*)-J(P_n,Q,\gamma_n^*),J(P_n, Q, \gamma^*)-J(P,Q,\gamma^*)).
\end{align*}
    
    As $c$ is bounded it follows that for any $\gamma'\in\Gamma$,
    \begin{align}\label{eq:TVbound}
      \begin{split}
      |J(P_n,Q,&\gamma')-J(P,Q,\gamma')|\\
      &=\left|\int c(x,\gamma'(y))P_nQ(\dd x,\dd y)-\int c(x,\gamma'(y))PQ(\dd x,\dd y)\right|\\
      &\leq \|c\|_\infty \|P_nQ-PQ\|_{TV}=\|c\|_\infty \|P_n-P\|_{TV},
      \end{split}
    \end{align}
    where we have used Lemma~\ref{lemma:joint} for the last equality. Inequalities \eqref{eq:max} and \eqref{eq:TVbound} together imply that $|J^*(P,Q)-J^*(P_n,Q)|\leq \|c\|_\infty \|P_n-P\|_{TV} $. Since $\|P_n-P\|_{TV}\to0$  we have that $J^*(P_n,Q)\to J^*(P,Q)$.
  \end{proof}

\adk{
\begin{remark}
In this paper we only focus on the case where the channel is known by the controller. That is the true channel model $Q$ is available to the controller. For the case where this is no longer true, some further analysis is required. If prior model $P$ and the channel $Q$ are not known, controller can have an estimating sequence $P_nQ_n \in \P(\mathds{X}\times \mathds{Y})$ for the true joint measure $PQ\in \P(\mathds{X}\times\mathds{Y})$. Now, the question becomes analyzing the convergence of $P_nQ_n \to PQ$. This joint convergence might require different set of assumptions on $P_n$ and $Q_n(\cdot|x)$ which we do not discuss on this paper. However, Lemma \ref{lemma:joint} might give an idea on this joint convergence where we consider the convergence of joint measure $P_nQ$ to $PQ$. In \cite{yuksel12:siam}, similar joint convergence is studied for convergence of measurement channels and fixed prior distributions. The reader can also refer to \cite[Lemma 2.2]{yuksel12:siam} for the analysis on convergence of $PQ_n \to PQ$.
\end{remark}
}

\subsection {Multi-stage and infinite-horizon discounted setup}
We now consider continuity problems for the multi stage case. For this case, our focus will be on the infinite stage discounted cost setting. Clearly, the lack of continuity for single-stage problems implies the lack of continuity of multi-stage problems. In the following, the emphasis will be on developing setups where continuity can be established. 

In particular, if we put further restrictions on the system model and the measurement channel, we will establish sufficient conditions for continuity under weak convergence of the priors. This will be studied in the following.

\subsubsection{Weak convergence}

Consider a partially observed Markov decision process (POMDP), with state space $\sX$, action space $\sA$, and observation space $\sY$, all Borel spaces.
Define the history spaces $\sH_{t}=(\sY \times \sU)^{t}\times\sY$, $t=0,1,2,\ldots$ endowed with their
product Borel $\sigma$-algebras generated by $\B(\sY)$ and $\B(\sU)$. A \emph{policy} $\pi=\{\pi_{t}\}$ is a sequence of stochastic kernels on $\sU$ given $\sH_{t}$. We denote by $\Pi$ the set of all policies. For any initial distribution $\mu$ and policy $\pi$ we can think of the POMDP as a stochastic process $\bigl\{ X_t,Y_t,U_t \bigr\}_{t\geq0}$ defined on the probability space $\bigl( \Omega, \B(\Omega), P_{\mu}^{\pi} \bigr)$, where $\Omega = \sH_{\infty} \times \sX^{\infty}$, the $X_t$ are $\sX$-valued random variables, the $Y_t$ are $\sY$-valued random variables, the $U_t$ are $\sU$-valued random variables.

It is known that any POMDP can be reduced to a (completely observable) MDP \cite{Yus76}, \cite{Rhe74}, whose states are the posterior state distributions or "beliefs" of the observer; that is, the state at time $t$ is
\begin{align}
Z_t(\,\cdot\,) := \sPr\{X_{t} \in \,\cdot\, | Y_0,\ldots,Y_t, U_0, \ldots, U_{t-1}\} \in \P(\sX). \nonumber
\end{align}
We call this equivalent MDP the belief-MDP \index{Belief-MDP}. The belief-MDP has state space $\sZ = \P(\sX)$ and action space $\sU$. Recall that $\sZ$ is equipped with the Borel $\sigma$-algebra generated by the topology of weak convergence \cite{Billingsley}. Since $\sX$ is a Borel space, $\sZ$ is metrizable with the Prokhorov metric which makes $\sZ$ into a Borel space \cite{Par67}. The transition probability $\eta$ of the belief-MDP can be constructed as follows (see also \cite{Her89}). If we define the measurable function $F(z,a,y) := \Pr\{X_{t+1} \in \,\cdot\, | Z_t = z, U_t = u, Y_{t+1} = y\}$ from $\sZ\times\sA\times\sY$ to $\sZ$ and the stochastic kernel $H(\,\cdot\, | z,u) := \Pr\{Y_{t+1} \in \,\cdot\, | Z_t = z, U_t = u\}$ on $\sY$ given $\sZ\times\sU$, then $\eta$ can be written as
\begin{align}
\eta(\,\cdot\,|z,u) = \int_{\sY} 1_{\{F(z,u,y) \in \,\cdot\,\}} H(dy|z,u). \nonumber
\end{align}
The one-stage cost function $c$ of the belief-MDP is given by
\begin{align}
\tilde{c}(z,u) := \int_{\sX} c(x,u) z(dx). \label{weak:eq8}
\end{align}
Hence, the belief-MDP is a Markov decision process with the components $(\sZ,\sU,\eta,\tilde{c})$.

%
It is a standard result that an optimal control policy will use the belief $z_t$ as a sufficient statistic for optimal policies (see \cite{Yus76}, \cite{Rhe74}). 

\begin{assumption}
\label{weak:as3}
\begin{itemize}
\item [(a)] The stochastic kernel $\mathcal{T}(dx_1|x_0=x,u_0=u)$ is weakly continuous in $(x,u)$.
\item [(b)] Assumption \ref{assumptionTV} holds; that is the observation channel $Q(dy|x)$ is continuous in total variation.
\item[(c)] The stage-wise cost function $c(x,u)$ is non-negative, bounded and continuous on $\mathds{X} \times \mathds{U}$.
\item [(d)] $\mathds{U}$ is compact. 
\end{itemize}
\end{assumption}

By \cite[Proposition 7.30]{bertsekas78}, the one stage cost function $\tilde{c}$ of the belief-MDP, which is defined in (\ref{weak:eq8}), is continuous and bounded, that is in $C_b(\sZ\times\sU)$, under Assumption \ref{weak:as3}-(a),(b). 
The following theorem is from \cite[Theorem 3.7, Example 4.1]{FeKaZg14} and \cite[Example 2.1]{saldi2014near}.
\begin{theorem}\label{weak:thm6}
\begin{itemize}
\item [(i)] Under Assumption \ref{weak:as3}, the stochastic kernel $\eta$ for belief-MDP is weakly continuous in $(z,u)$.
\item [(ii)] If we relax the continuity in total variation of the observation channel to weak continuity, then $\eta$ may not be weakly continuous even if the transition probability $p$ of POMDP is continuous in total variation.
\item [(iii)] $\eta$ may not be setwise continuous in $u$ even if the observation channel is continuous in total variation.
\end{itemize}
\end{theorem}

For an infinite horizon discounted cost problem with bounded costs, an optimal policy can be computed through the iterated use of the {\it discounted cost optimality operator} to be introduced below. Note that under Assumption~\ref{weak:as3}(c), by an application of the dominated convergence theorem $\tilde{c}$ is continuous and bounded. Now, for a MDP with weakly continuous transition probabilities and compact action spaces, it follows that an optimal control policy exists and that the optimal cost is continuous in the initial state (or probability measure in the context here): This follows because the {\it discounted cost optimality operator} $T: C_b(\sZ) \to C_b(\sZ)$ (see e.g. \cite[Chapter 8.5]{hernandezlasserre1999further}):
\[\bigg(T(v)\bigg)(z) = \min_{u} \bigg( \tilde{c}(z,u) + \beta E[v(z_1) | z_0=z, u_0=u] \bigg)\]
is a contraction from $C_b(\sZ)$ to itself under the supremum norm. As a result, there exists a fixed point, which is {\it continuous}. This fixed point is the value function. 

This argument shows that the value function is continuous in the belief state, $z_0(Y)(x \in\cdot) = P(x\in\cdot|Y)$, which is the {\it posterior} distribution of the state variable given the observations. However, convergence of the distribution of the priors may not always imply the convergence of the posteriors: For the single stage case, with the cost function of the belief process $\tilde{c}$ defined as in (\ref{weak:eq8}), the value function is given by $J^*(z)=\inf_{u \in \mathds{U}}\tilde{c}(z,u)$. It can be seen that the value function is again continuous in the belief state $z$ if $c$ is continuous in $x$ and $\mathds{U}$ is compact (by an application of the dominated convergence theorem). However, as we have seen in the counterexample used to prove Theorem \ref{weak_count}; even though $P_n \to P$ weakly, value functions do not converge. The next theorem shows that with further conditions, convergence of the posteriors (belief states) can also be guaranteed. Before the main result we first present a key lemma.

\begin{lemma}\label{cost_weak}
Under Assumption \ref{weak:as3}, as $P_n \to P$ weakly,
\begin{align*}
\sup_{\gamma \in \Gamma}\big|E_{P}\big[c(X_k,\gamma(Y_{[0,k]}))\big] - E_{P_n}\big[c(X_k,\gamma(Y_{[0,k]})) \big] \big| \to 0
\end{align*}
for any time stage $k<\infty$.
\end{lemma}

\begin{proof}
The proof can be found in the Appendix \ref{lem_212}.
\end{proof}

\begin{theorem}\label{ContWeakConv}
Suppose that Assumption \ref{weak:as3} holds. Then, as $P_n \to P$ weakly $|J_{\beta}^*(P_n,{\cal T})-J_{\beta}^*(P,{\cal T})| \to 0$. 
\end{theorem}
\begin{proof}
We start with the following bound;
\begin{align}\label{main_bnd}
      \begin{split}
        |J_\beta^*(P,\mathcal{T}&)-J_\beta^*(P_n,\mathcal{T})|\\
        &\leq\max(J_\beta(P, \mathcal{T}, \gamma_n^*)-J(P_n,\mathcal{T},\gamma_n^*),J_\beta(P_n, \mathcal{T}, \gamma^*)-J(P,\mathcal{T},\gamma^*)).
      \end{split}
    \end{align}

Now, we try to show that under Assumption \ref{weak:as3}, as $P_n \to P$ weakly,
\begin{align}\label{weak_same_policy}
\sup_{\gamma \in \Gamma}|J_{\beta}(P_n,\mathcal{T},\gamma) - J_{\beta}(P,\mathcal{T},\gamma)| \to 0.
\end{align}

To prove (\ref{weak_same_policy}) we start with the following inequality:
\begin{align*}
\sup_{\gamma \in \Gamma}&|J_{\beta}(P_n,\mathcal{T},\gamma) - J_{\beta}(P,\mathcal{T},\gamma)| \\
&\leq\sup_{\gamma \in \Gamma}\sum_{k=0}^{T} \beta^k\bigg|E_P\big[c(X_k,\gamma(Y_{[0,k]}))\big] - E_{P_n}\big[c(X_k,\gamma(Y_{[0,k]}))\big]\bigg| + \sum_{k=T}^{\infty}\beta^k2\|c\|_\infty\nonumber.
\end{align*}
First we fix $\epsilon >0$ and find a $T_\epsilon$ such that $\sum_{k=T_\epsilon}^{\infty}\beta^k2\|c\|_\infty\leq\epsilon /2$. Now, we claim that we can find an $N$  such that for every $n>N$ 
\begin{align*}
&\sup_{\gamma \in \Gamma} \sum_{k=0}^{T_\epsilon -1}\bigg|E_{P_n}\big[c(X_k,\gamma(Y_{[0,k]}))\big] - E_P\big[c(X_k,\gamma(Y_{[0,k]}))\big]\bigg|\\
&\leq \sum_{k=0}^{T_\epsilon -1}\sup_{\gamma \in \Gamma}\bigg|E_{P_n}\big[c(X_k,\gamma(Y_{[0,k]}))\big] - E_P\big[c(X_k,\gamma(Y_{[0,k]}))\big]\bigg| \leq\epsilon/2.
\end{align*}

Lemma \ref{cost_weak} implies that for every time stage $k<T_\epsilon$,  $\sup_{\gamma \in \Gamma}\bigg|E_{P}\big[c(X_k,\gamma(Y_{[0,k]}))\big] - E_{P_n}\big[c(X_k,\gamma(Y_{[0,k]})) \big] \bigg|$ can be made less than $\epsilon/2T_\epsilon$ for all $n>N$ for some $N<\infty$. Since $T_\epsilon <\infty$, we can find a common $N$ for all time stages $k<T_\epsilon$. So we can write
\begin{align*}
\sup_{\gamma \in \Gamma}\sum_{k=0}^{T_\epsilon-1} \beta^k \bigg|E_{P}\big[c(X_k,\gamma(Y_{[0,k]}))\big] - E_{P_n}\big[c(X_k,\gamma(Y_{[0,k]})) \big] \bigg|< \epsilon/2
\end{align*}

Combining the results,  we have
\begin{align*}
\sup_{\gamma \in \Gamma}|J_{\beta}(P_n,\mathcal{T},\gamma) - J_{\beta}(P,\mathcal{T},\gamma)| &\leq  \sum_{k=0}^{T_\epsilon-1} \beta^k \bigg|E_{P}\big[c(X_k,\gamma(Y_{[0,k]}))\big] - E_{P_n}\big[c(X_k,\gamma(Y_{[0,k]})) \big] \bigg| \\
&\quad \quad +\sum_{k=T_\epsilon}^\infty\beta^k 2\|c\|_\infty\\
&<\epsilon
\end{align*}
for all $n>N$ for some $N<\infty$ for every given $\epsilon>0$, which proves (\ref{weak_same_policy}).

 Now looking at the term in (\ref{main_bnd}):
\[\max(J_\beta(P, \mathcal{T}, \gamma_n^*)-J(P_n,\mathcal{T},\gamma_n^*),J_\beta(P_n, \mathcal{T}, \gamma^*)-J(P,\mathcal{T},\gamma^*)),\]
we can see that both terms go to zero using (\ref{weak_same_policy}).
\end{proof}

In the following we give a weaker result, which holds with no restrictions, however.

\subsubsection{Continuity under total variation and strategic measures}
For stochastic control problems, {\it strategic measures} are defined (see Sch\"al \cite{Schal}, also \cite{dynkin1979controlled,feinberg1996measurability}) as the set of probability measures induced on the product spaces of the state and action pairs by measurable control policies: Given an initial distribution on the state, and a policy, one can uniquely define a probability measure on the infinite product space consistent with finite dimensional distributions, by Ionescu Tulcea theorem \cite{HernandezLermaMCP}. Now, define a {\it strategic measure} under a policy $\gamma^n= \{\gamma^n_0,\gamma^n_1, \cdots, \gamma^n_k,\cdots\}$ as a probability measure defined on ${\cal B}(\mathds{X} \times \mathds{Y} \times \mathds{U})^{\mathds{Z}_+}$ by:
\begin{align*}
&P^{\gamma^n}_{P_n}(\dd(x_0,y_0,u_0),\dd(x_1,y_1,u_1),\cdots) \nonumber \\
& \quad = P_n(\dd x_0) Q(\dd y_0|x_0) 1_{\{\gamma^n(y_0) \in \dd u_0\}} \mathcal{T}(\dd x_1|x_0,u_0) Q(\dd y_1|x_1) 1_{\{\gamma^n(y_0,y_1) \in \dd u_1\}} \cdots
\end{align*}
Under a strategic measure $P^{\gamma^n}_{P_n}$ we define, 
\begin{align*}
  J^*_{\beta}(P_n,\mathcal{T}) = \inf_{\gamma^n}E^{\gamma^n}_{P_n}\left[\sum_k \beta^k c(x_k,\gamma^n_k(y_{[0,k]}))\right] 
\end{align*}

\begin{theorem}\label{multiStage}
If $c(x,u)$ is a non-negative, measurable and bounded function in $\mathds{X}\times\mathds{U}$ then
\[|J^*_{\beta}(P_n,\mathcal{T}) - J^*_{\beta}(P,\mathcal{T})| \leq  \|P_n(x_0 \in \cdot)-P(x_0 \in \cdot) \|_{TV}  \frac{1}{1 - \beta}  \|c\|_\infty. \]
\end{theorem}

\begin{proof}
From inequalities \eqref{eq:ineq1}, \eqref{eq:ineq2} and (\ref{eq:max}) we have that $|J^*_{\beta}(P_n,\mathcal{T})-J^*_{\beta}(P,\mathcal{T})|$ is upper bounded as follows,
\begin{align*}
  &|J^*_{\beta}(P_n,\mathcal{T})-J^*_{\beta}(P,\mathcal{T})|\nonumber \\
  &\quad \leq \max\bigg(\sum_k \beta^k  \|P^{\gamma^n}_{P}(x_k \in \cdot, y_{[0,k]}\in\cdot) - P^{\gamma^n}_{P_n}(x_k \in\cdot,  y_{[0,k]}\in\cdot) \|_{TV} \sup c(x,\gamma^n(y_{[0,k]})), \nonumber \\
  & \quad\quad\quad \quad \quad \sum_k \beta^k  \|P^{\gamma}_{P_n}(x_k\in\cdot, y_{[0,k]}\in\cdot) - P^{\gamma}_{P}( x_k\in\cdot, y_{[0,k]}\in\cdot) \|_{TV} \sup c(x,\gamma(y_{[0,k]})\bigg) \nonumber 
  \end{align*}
  For any $\gamma \in \Gamma$, we have,
  \begin{align}
&\sum_k \beta^k  \|P^{\gamma}_{P_n}(x_k\in\cdot, y_{[0,k]}\in\cdot) - P^{\gamma}_{P}(x_k \in\cdot, y_{[0,k]}\in\cdot) \|_{TV} \sup c(x,\gamma(y_{[0,k]})) \nonumber \\
  &\quad = \sum_k \beta^k  \|P^{\gamma}_{P}(x_k\in\cdot, y_{[0,k]}\in\cdot) - P^{\gamma}_{P_n}(x_k\in\cdot, y_{[0,k]}\in\cdot) \|_{TV} \|c\|_\infty \nonumber \\
  &\quad \leq  \frac{1}{1 - \beta}  \|c\|_\infty \|P_n(x_0\in\cdot)-P(x_0\in\cdot)\|_{TV}.
\end{align}
Here, we use the property that 
\begin{align}
 & \|P^{\gamma}_{P_n}(x_k\in\cdot, y_{[0,k]}\in\cdot) - P^{\gamma}_P(x_k\in\cdot, y_{[0,k]}\in\cdot) \|_{TV} \nonumber \\
 & \quad \leq \|P^{\gamma}_{P_n}((x,y,u)_{[0,\infty)}\in\cdot) - P^{\gamma}_P((x,y,u)_{[0,\infty)}\in\cdot) \|_{TV} =  \|P_n(x_0\in\cdot)-P(x_0\in\cdot) \|_{TV},  \nonumber
\end{align}
similar to the derivation in Lemma \ref{lemma:joint}.
\end{proof}

\subsection{Some remarks on the infinite horizon average cost setup}
In this section, we show that the conditions presented earlier may not lead to continuity for the infinite horizon setup but alternative conditions will likely be useful. Consider an infinite horizon average cost setup with the  objective function given by
\begin{align*}
    J_{\infty}(P,\mathcal{T},\gamma)= \limsup_{T \to \infty} {\frac{1}{T}}E_P^{Q,\gamma}\left[\sum_{t=0}^{T-1}c(X_t,U_t)\right].
  \end{align*}
The optimal cost is given by
\[J_{\infty}^*(P,\mathcal{T})=\inf_{\gamma\in\Gamma}J_{\infty}(P,\mathcal{T},\gamma).\]

Given an optimal stationary policy, whose existence follows from the conditions given in Assumption \ref{weak:as3}, say through the convex analytic method (see Borkar \cite{Borkar2}), the process process $\pi_k$ becomes Markovian. It is known that (see \cite[Theorems~2.3.4-2.3.5]{Hernandez}) if this belief Markov process admits a unique invariant measure then
\[\frac{1}{N}E_{\pi_0}\big[\sum_{k=0}^{N-1}c(\pi_k)\big] \to c\] 
almost surely for all initial conditions (that is: 'priors') for some constant $c$. Thus, if the belief process has a unique invariant measure, the continuity in priors holds immediately as the average cost does not depend on the prior measure provided that the initial prior belongs to the support set of the invariant measure. However, checking ergodicity is a challenging problem for the controlled setup; this is a subject of current research.

On the other hand, we show in the following that Assumption \ref{weak:as3} does not alone guarantee continuity.


\begin{exmp}
Let $\mathds{X} = \mathds{R}$, $\mathds{U}=[-1,1]$, $\mathds{Y}=[-1,1]$.
Suppose we are given two initial distributions, the transition kernel and the measurement channel as
\begin{align*}
P(\cdot) &= \delta_0(\cdot)\\
P_n(\cdot) &= \frac{1}{2}\delta_{1/n}(\cdot) + \frac{1}{2}\delta_{-1/n}(\cdot)\\
\mathcal{T}(\cdot|x,u)&=\delta_{2x}(\cdot)\\
Q(\cdot|x)&=U[-1,1].
\end{align*}
Notice that the transition kernel is weakly continuous in $(x,u)$, the measurement channel is continuous in total variation and $P_n\to P$ weakly. 

The stage-wise cost function is defined by as follows.
\begin{align*}
c(x,u)=&\begin{cases}(x+u)^2 \quad \text{ if } |x|\leq1\\
(1+u)^2  \quad \text{ if } |x|>1\end{cases}
\end{align*}
So the cost is always bounded. The optimal control actions for both initial distributions at any time $k\geq0$ are $\gamma(y_{[0,k]})=0$ and $\gamma_n(y_{[0,k]})=0$.

It is easy to see that the optimal cost for $P$ is 0. The optimal cost for $P_n$ can be calculated as follows.
\begin{align*}
&J_\infty^*(P_n,\mathcal{T})=\lim_{N \to \infty}\frac{1}{N}\bigg(\sum_{k=0}^{\log_2 n}\big(\frac{2^k}{n}\big)^2 + \sum_{k=\log_2 n +1}^{N}(1)^2\bigg)\\
&=\lim_{N \to \infty}\frac{1}{N}\bigg(\frac{4(n^2)}{3}-\frac{1}{3} +N-\log_2n +1\bigg)=1 \neq 0
\end{align*}
\end{exmp}
A more complete treatment for the average cost case will be reported in future work.

\section{Robustness}\label{robustness}

\subsection{Robustness to incorrect priors and mismatch bounds for single stage problems}

First, for this subsection, we will consider the single stage stochastic control problem with cost function $c$, initial distribution $P$. We shall denote this problem by $\Xi=(c,P)$. Consider the following problem: let $\widetilde{P}$ be another initial distribution. Decision maker $DM$ computes an optimal policy, $\widetilde{\gamma}^*$, for the problem ${\widetilde{\Xi}=(c,\widetilde{P})}$ and applies it to $\Xi$. Can we approximate the loss in performance, that is, can we find a bound on $J(P,Q,\widetilde{\gamma}^*)-J^*(P,Q)$? This situation naturally arises when the initial distribution, $P$, is uncertain and $DM$ has a prior belief, $\widetilde{P}$, which is perhaps based on an incorrect initial model.


 \begin{prop}\label{prop:tvbound}
    Assume that $c:\mathds{X}\times\mathds{U}\to\R$ is nonnegative, measurable, and bounded. Let $\widetilde{\gamma}^*$ be an optimal (or $\varepsilon$-optimal) control policy for the single stage stochastic control problem $\widetilde{\Xi}=(c,\widetilde{P},Q)$, where $c$ is a cost function, $\widetilde{P}$ is an initial distribution, and $Q$ is a measurement channel. Let $P$ be another probability distribution on $\mathds{X}$. Then $|J(P,Q,\widetilde{\gamma})-J^*(P,Q)|\leq2\|c\|_\infty \|P-\widetilde{P}\|_{TV}$.
  \end{prop}
  \begin{proof}
    Let $\gamma^*$ denote an optimal (or $\varepsilon$-optimal) policy for the problem $\Xi=(c,P,Q)$. We have,
    \begin{align*}
      |J(P,Q,\widetilde{\gamma}^*)-J^*(P,Q)|&=|J(P,Q,\widetilde{\gamma}^*)-J(\widetilde{P},Q,\widetilde{\gamma}^*)+J(\widetilde{P},Q,\widetilde{\gamma}^*)-J(P,Q,\gamma^*)|\\
      &\leq|J(P,Q,\widetilde{\gamma}^*)-J(\widetilde{P},Q,\widetilde{\gamma}^*)|+|J^*(\widetilde{P},Q)-J^*(P,Q)| ~(+2\varepsilon)\\
      &\leq2\|c\|_\infty \|P-\widetilde{P}\|_{TV},
    \end{align*}
    where we have used Equation \eqref{eq:TVbound} for the final inequality.
  \end{proof}

Similar to the continuity section, for the weak convergence of priors we have a negative result. The following result says that the mismatch error may not diminish even if our belief model converges weakly to the true model.
 \begin{prop}\label{prop:SWextension}
    Let an initial distribution $P\in\mathcal{P}(\mathds{X})$, and a cost function $c:\mathds{X}\times\mathds{U}\to \R$ be given. Assume that $P_n \to P$ weakly and let $\gamma_{P_n}^*$ be an optimal (or $\varepsilon$-optimal) policy for the control problem $\Xi_n=(c,P_n)$. It does not follow that $J(P,Q,\gamma_{P_n}^*)\to J^*(P,Q)$ as $n\to\infty$. This result holds even if $c$ is bounded and continuous in both $x$ and $u$.
  \end{prop}

\begin{proof}
    We use an extension of the proof from Theorem~\ref{ex:setwise}. Let $\mathds{X}=\mathds{Y}=\mathds{U}=[0,1]$ and let $c(x,u)=(x-u)^2$. We let $P_n \sim f_n$ for $n\in\N$ as in the proof of Theorem~\ref{ex:setwise}. We have that $P_n\to P \sim U([0,1])$ setwise. We let the channel be
    \begin{align*}
      Q(\cdot|x)\sim\frac{1}{2}\cdot\delta_x+\frac{1}{2}\cdot U([0,1]).
    \end{align*}
    As presented in the proof of Theorem~\ref{ex:setwise}, the optimal policy for the control problem $\Xi_n=(c,P_n,Q)$ is given by,
    \begin{align*}
      \gamma_{P_n}^*(y)=
      \begin{cases}
        \frac{1}{2}-\frac{1}{4n} & \text{if } y\in\cup_{k=1}^n R_{n,k}\\
        \frac{1}{3}\cdot\left(\frac{1}{2}-\frac{1}{4n}\right)+\frac{2}{3}y & \text{if } y\in\cup_{k=1}^n L_{n,k}
      \end{cases}.
    \end{align*}
    If the decision maker applies $\gamma_{P_n}^*$ to the control problem $\Xi=(c,P,Q)$, this results in the following cost (the calculations can be found in \ref{app:SWextension}):
    \begin{align*}
      J(P,Q,\gamma_{P_n}^*)&=\frac{2}{27}+\frac{5}{72n^2}.
    \end{align*}
    We notice that as $n\to\infty$, $|J(P,Q,\widetilde{\gamma}^*)-J^*(P,Q)|\to\frac{2}{27}-\frac{1}{16}=\frac{5}{432}\neq0$.
  \end{proof}

We now present a positive result for weak convergence. 
 
\begin{theorem}\label{thm:weaksuffRobust}
Suppose that $c(x,u)$ is bounded, $\mathds{U}$ is compact and Assumption \ref{assumptionTV} holds. Then, as $P_n \to P$ weakly $\lim_{n\to\infty}|J(P,Q,\gamma_{P_n}^*)-J^*(P,Q)| \to 0$, that is the system is robust to errors in the priors under weak convergence.
\end{theorem}

\begin{proof}
  \begin{align}
    &\lim_{n\to\infty}|J(P,Q,\gamma_{P_n}^*)-J^*(P,Q)|\nonumber \\
    &\quad\leq\lim_{n\to\infty}|J(P,Q,\gamma_{P_n}^*)-J(P_n,Q,\gamma_{P_n}^*)| + \lim_{n\to\infty}|J(P_n,Q,\gamma_{P_n}^*)-J(P,Q)| \label{eq:weakrobust}\\
    &\quad=0.\nonumber
  \end{align}
  We note that the first term goes to 0 with the same argument used in \ref{maxterm} and the second terms goes to zero by Theorem~\ref{thm:weaksuff}. We conclude that under the given assumptions, the control problem is robust under weak convergence.
 \end{proof} 

\subsection{Robustness to incorrect priors for multi stage problems}


The following result holds in generality. 

\begin{theorem}
If $c(x,u)$ is a non-negative, measurable and bounded function in $\mathds{X}\times\mathds{U}$ then
\[|J_{\beta}(P,Q,\gamma_{P_n}^*)-J_{\beta}^*(P,Q)| \leq 2 \|P_n(\dd x_0)-P(\dd x_0) \|_{TV} \frac{1}{1 - \beta}  \|c\|_\infty. \]
\end{theorem}

\begin{proof}
We use that 
\[J_{\beta}(P,Q,\gamma_{P_n}^*)-J_{\beta}^*(P,Q) = J_{\beta}(P,Q,\gamma_{P_n}^*) - J_{\beta}(P_n,Q,\gamma_{P_n}^*) + J_{\beta}(P_n,Q,\gamma_{P_n}^*) - J_\beta^*(P,Q).\]
From inequalities \eqref{eq:ineq1}, \eqref{eq:ineq2} and (\ref{eq:max}) we have that $|J_{\beta}(P,Q,\gamma_{P_n}^*) - J_{\beta}(P_n,Q,\gamma_{P_n}^*)|$ is upper bounded as $ \|P_n(\dd x_0)-P(\dd x_0) \|_{TV} \frac{1}{1 - \beta}  \|c\|_\infty $.
The analysis is then complete by considering Theorem \ref{multiStage}.
\end{proof}

 We now develop a robustness result under weak convergence of priors for multi-stage case. First, we give a lemma showing that for any multi-stage setting with a controlled Markov chain satisfying Assumption \ref{weak:as3}, the cost at any time stage is continuous in priors under weak convergence.

\begin{theorem}\label{weak_robust}
Under Assumption \ref{weak:as3}, as $P_n \to P$ weakly, we have,
\[|J_{\beta}(P,\mathcal{T},\gamma_{P_n}^*)-J_{\beta}^*(P,\mathcal{T})| \to 0\]
\end{theorem}

\begin{proof}
We use the following bound again,
\[|J_{\beta}(P,\mathcal{T},\gamma_{P_n}^*)-J_{\beta}^*(P,\mathcal{T})| \leq |J_{\beta}(P,\mathcal{T},\gamma_{P_n}^*) - J_{\beta}(P_n,\mathcal{T},\gamma_{P_n}^*)| + |J_{\beta}(P_n,\mathcal{T},\gamma_{P_n}^*) - J_\beta^*(P,\mathcal{T})|.\]
Here, $\gamma_{P_n}^*$ is optimal for prior $P_n$, the existence of $\gamma_{P_n}^*$, and an optimal policy $\gamma$ for prior $P$ follows from Theorem \ref{ContWeakConv}. The second term goes to zero by Theorem \ref{ContWeakConv}. The first term goes to zero by (\ref{weak_same_policy}) which states that
\begin{align*}
\sup_{\gamma \in \Gamma}|J_{\beta}(P_n,\mathcal{T},\gamma) - J_{\beta}(P,\mathcal{T},\gamma)| \to 0.
\end{align*}
\end{proof}

\section{Implications for Empirical Learning Methods in Stochastic Control}\label{empiricalLearning}
In engineering practice, when one does not know the probability measure for a random variable one typically attempts to learn it via test inputs
or empirical observations.  Let $\{(X_i), i \in \mathbb{N} \}$ be an $\mathbb{X}$-valued i.i.d random variable
sequence generated according to some distribution $\mu$. 

Defining for every (fixed) Borel $B \subset \mathbb{X}$, and $n \in
\mathbb{N}$, the empirical occupation measures
\[
\mu_n(B)=
\frac{1}{n}\sum_{i=1}^{n} 1_{\{X_i \in B\}},
\]
one has $\mu_n(B) \to \mu(B)$ almost surely (a.s$.$) by the strong law
of large numbers. Also, $\mu_n \to \mu$ weakly with probability one (\cite{Dudley02}, Theorem 11.4.1).

However, $\mu_n$ can not converge to $\mu$ in total variation, in general. On the other hand, if we know that $\mu$ admits a density, we can find estimators to estimate $\mu$ under total variation \cite{Devroye85}.

As discussed above, the empirical averages converge almost surely. By a similar reasoning, for a given bounded measurable function $f$, $\int \mu_n(dx) f(x)$ converges to $\int \mu(dx) f(x)$. This then also holds for any {\it finite} collection of functions, $f_1, \cdots, f_n$ for some $n \in \mathbb{N}$. A relevant question is the following: Can one ensure uniform convergence (over a family of functions) with arbitrary precision by only guaranteeing convergence for a finite collection of functions. This entails the problem of covering a family of functions with arbitrarily small neighborhoods of finitely many functions under an appropriate distance metric. The answer to this question is studied by the theory of {\it empirical risk minimization}: In the learning theoretic context when one tries to estimate the source distribution, the convergence of optimal costs under $\mu_n$ to the cost optimal for $\mu$ is called the {\it consistency of empirical risk minimization} \cite{Vap00}. 

In particular, if the following uniform convergence holds,
\begin{align}\label{unif_conv}
&\lim_{n \to \infty} \sup_{f \in {\cal F}} \bigg|\int
f(x) \mu_n(dx) - \int f(x) \mu(dx) \bigg|
=0,
\end{align}
for a class of measurable functions ${\cal F}$, then ${\cal F}$ is called a {\em $\mu$-Glivenko-Cantelli class} \cite{Dudley}. If the class ${\cal F}$ is $\mu$-Glivenko-Cantelli for every $\mu$, it is called a {\it universal Glivenko-Cantelli} class. One example of a universal Glivenko-Cantelli family of real functions on $\mathbb{R}^N$ is the
family $\{f:\, \|f\|_{BL}\le M\}$ for some $0<M<\infty$, where
$\|f\|_{BL}=\|f\|_\infty + \sup_{x_1\neq x_2}\frac{|f(x_1)-f(x_2)|}{|x_1-x_2|}$ ( \cite{Dudley}). For related characterizations and further examples, see \cite{raginsky2013empirical} \cite{VanHandel} \cite{dudley1999uniform}.

In another direction, if (\ref{unif_conv}) holds for any sequence of measures $\{\mu_n\}$ converging weakly to $\mu$ (rather than only empirical models), then ${\cal F}$ is called a  $\mu$-uniform class. For a subset, $B$, of $\mathds{X}$, the oscillation of ${\cal F}$ on $B$ is defined as
\[ w_{\cal F}(B)=\sup \{|f(x)-f(y)|, f \in {\cal F}, x,y \in B \},\]
in the case where $\cal F$ consists of a single function $f$ we use the notation $w_f (B)$ or $w_f B$. Then a characterization for uniformity classes is given by the following \cite{billingsley1967uniformity}: a necessary and sufficient condition for ${\cal F}$ to be a $P$-uniformity class is that 
\[  w_{\cal F}(\mathds{X}) < \infty\]
and
\[ \lim_{\delta \to 0} \sup_{f \in {\cal F}} P\{x: w_fS(x,\delta) > \epsilon\} =0\] for any $\epsilon>0$, where $S(x,\delta)$ is the ball around $x$ with radius $\delta$.
It can be seen that one example of a uniformity class on $\R^N$ is again the
family $\{f:\, \|f\|_{BL}\le M\}$ for some $0<M<\infty$.
 For a detailed discussion and characterization of these classes see \cite{billingsley1967uniformity}.


\subsection{Application to robustness to incorrect priors}
$\quad$


\begin{theorem}
Suppose that the prior model of the system is estimated with the i.i.d. measurements such that for every (fixed) Borel $B \subset \mathbb{X}$, and $n \in
\mathbb{N}$, the empirical occupation measures
\begin{align}\label{empir}
P_n(B)=
\frac{1}{n}\sum_{i=1}^{n} 1_{\{X_i \in B\}},
\end{align}
then under one of the following conditions,
 \begin{itemize}
\item[(i)]
if the stage-wise cost function $c(x,u)$ is non-negative, bounded and continuous on $\mathds{X}\times\mathds{U}$, the action space $\mathds{U}$ is compact and the measurement channel $Q(\cdot|x)$ is continuous in total variation,
\item[(ii)]
If we restrict the class of policies to ${\cal G}=\{\gamma: \int Q(dy|x)c(x,\gamma(y)) \in {\cal E}\},$ where ${\cal E}$ is a class of $P$-Glivenko-Cantelli family of functions,
\item[(iii)]
if we restrict the class of policies to ${\cal G}=\{\gamma: c(x,\gamma(y)) \in {\cal E}\},$ where ${\cal E}$ is a class of $PQ$-uniformity family of functions and if the measurement channel $Q(\cdot|x)$ is weakly continuous in $x$,
\end{itemize}
as $n \to \infty$, with probability 1, we have, $J^*(P_n,Q)\to J^*(P,Q)$ and $J(P,Q,\gamma_n^*) \to J^*(P,Q)$, where $\gamma_n^*$ is the optimal policy designed for the estimated model $P_n$. That is the optimal cost function and the optimal policies are consistent under empirical estimation.
\end{theorem}
\begin{proof}
 \begin{itemize}
\item[(i)]
The result follows from Theorem \ref{thm:weaksuff} and Theorem \ref{thm:weaksuffRobust}.
\item[(ii)] 
We write 
\begin{align*}
\lim_{n\to \infty} \sup_{\gamma \in\Gamma}\big|\int P_n(dx) \int Q(dy|x)c(x,\gamma(y))- \int P(dx) \int Q(dy|x)c(x,\gamma(y))\big|.
\end{align*}
The results follows from the discussion made for (\ref{unif_conv}).
\item[(iii)]
We make the following argument. 
Similar to the above discussion, if optimal
policies are assumed to be
from the restricted class of policies ${\cal G}$, a sufficient condition for the convergence of optimal costs is the following form of uniform weak convergence:
\begin{align*}
&\lim_{n \to \infty} \sup_{\gamma \in {\cal G}}
\bigg|\int_{\mathbb{X}\times \mathbb{Y} }
c(x,\gamma(y)) P_nQ(dx,dy)   \\
& \qquad\qquad\quad- \int_{\mathbb{X}\times \mathbb{Y} }
c(x,\gamma(y))  PQ(dx,dy)\bigg | = 0.
\end{align*}
The above argument holds when ${\cal G}=\{\gamma: c(x,\gamma(y)) \in {\cal E}\},$ where ${\cal E}$ is a class of $PQ$-uniformity family of functions and when the joint measure $P_nQ$ converges weakly to $PQ$. We have that the empirical measures $P_n$ converges weakly to $P$ with probability 1. Thus, the result follows from Lemma \ref{lemma:joint}(i).
\end{itemize}
\end{proof}

\adk{
Conditions (ii) and (iii) might be difficult to check in general. Therefore, we analyze these conditions from a practical point of view. Condition $(ii)$ requires $\int Q(dy|x)c(x,\gamma(y))$ to be a $P$-Glivenko-Cantelli class of functions. There are various characterizations of such class of functions (see \cite{Vap00}) One example of Glivenko-Cantelli class of functions is bounded Lipschitz functions. Thus, if $\int Q(dy|x)c(x,\gamma(y))$ is bounded and Lipschitz in $x$, the condition (ii) is satisfied. This requirement is met under the following set of restrictions
\begin{itemize}
\item $\|Q(\cdot|x)-Q(\cdot|x')\|_{TV}\leq \alpha |x-x'|$ for some $\alpha<\infty$ and $|c(x,u)-c(x',u)|\leq \beta |x-x'|$ for all $u \in \mathds{U}$ for some $\beta <\infty$.
\end{itemize}
The above condition is stronger than condition (i).

Condition (iii) requires $c(x,\gamma(y))$ to be a $PQ$-uniformity class of functions and $Q(\cdot|x)$ to be weakly continuous. One example of uniformity class of functions is again bounded Lipchitz functions. Thus, if $c(x,\gamma(y))$ is bounded and Lipschitz in $x$ and $y$ and if $Q$ is weakly continuous, the condition $(iii)$ is satisfied. The following is sufficient:
\begin{itemize}
\item $|c(x,u)-c(x',u')|\leq \alpha |x-x'|+\beta |u-u'|$ for some $\alpha<\infty$ and $\beta<\infty$. 
\item $\Gamma$ is restricted to be the space of Lipschitz functions such that for any $\gamma \in \Gamma$ we have $|\gamma(y)-\gamma(y')|\leq \eta|y-y'|$.
\item $Q(\cdot|x)$ is weakly continuous.
\end{itemize}
This set assumptions weakens the restrictions on the channel. However, putting assumptions on the space of policies is artificial as there is usually no guarantee that the optimal policy is Lipschitz. 
}

We next consider the multi-stage case.

\begin{theorem}
Suppose that the prior model of the system is estimated with the i.i.d. measurements as in (\ref{empir}). Under Assumption \ref{weak:as3}, as $n \to \infty$ we have almost surely, $J_\beta^*(P_n,{\cal T})\to J_\beta^*(P,{\cal T})$ and $J_\beta(P,{\cal T},\gamma_n^*) \to J_\beta^*(P,{\cal T})$, where $\gamma_n^*$ is the optimal policy designed for the estimated model $P_n$. That is, the optimal cost function and the optimal policies are consistent under empirical estimation for the multi stage problem.
\end{theorem}
\begin{proof}
Given that $P_n \to P$ weakly (almost surely), the result follows from Theorem \ref{ContWeakConv} and Theorem \ref{weak_robust}.
\end{proof}

\section{Conclusion}
We studied the topological properties of single
and multi stage optimization problems in stochastic control on the space of
initial probability measures, and applications of these 
to robustness of the control policies applied to systems with incomplete models. We made the observation that while weak convergence is in general too {\it weak} for continuity and robustness, channel and transition kernel regularities often allow for continuity and robustness under weak convergence. This is a practically very important result since often in engineering applications, system models are learned through training data which only guarantees weak convergence to the true model in general.

\section{Acknowledgements}
The authors are grateful to Graeme Baker who as an NSERC USRA summer student worked on the initial stages for some of the results reported here and the numerical analysis utilized in the proof of Theorem \ref{ex:setwise} presented in Section \ref{calculationsSetwise}. The authors are also grateful to Prof. Tam\'as Linder for the joint work in \cite{yuksel12:siam}, and Prof. Naci Saldi for many technical discussions with regard to Theorem \ref{weak:thm6}. 

  \appendix
  \section{Technical Proofs}\label{app1}

Here, we include some additional, mainly numerical, derivations utilized in the paper. The numerical derivations may or may not be included in the final version of the paper.

\subsection{Proof of Lemma \ref{cost_weak}}\label{lem_212}
\begin{proof}
Our goal is to show that for a given $\epsilon$ the term can be bounded by $\epsilon$ for $n>N$ for a sufficiently large $N$. For the ease of notation we will first study the case where $k=2$, then we will look at the general case.

In the following, to economize the notation, we write $\gamma(Y_{[0,t]})$ to denote $\{\gamma_k(Y_0,\cdots,Y_k; U_0,\cdots, U_{k-1}),  k \leq t\}$.
\begin{eqnarray}
&&\sup_{\gamma \in \Gamma}\bigg|E_{P}\big[c(X_2,\gamma(Y_{[0,2]}))\big] - E_{P_n}\big[c(X_2,\gamma(Y_{[0,2]})) \big] \bigg| \nonumber \\
&& = \sup_{\gamma \in \Gamma}\bigg| \int P(dx_0) Q(dy_0 | x_0)\mathcal{T}(dx_1|x_0,\gamma(y_0))Q(dy_1|x_1) \times \bigg(E[c(X_2,\gamma(Y_{[0,2]})) | x_0,x_1, y_{0},y_1 ] \bigg) \nonumber \\
&& \quad - \int P_n(dx_0) Q(dy_0|x_0) )\mathcal{T}(dx_1|x_0,\gamma(y_0))Q(dy_1|x_1) \times \bigg(E[c(X_2,\gamma(Y_{[0,2]})) | x_0, x_1, y_{0},y_1 ] \bigg)\bigg| \nonumber \\
&& = \sup_{\gamma \in \Gamma}\bigg| \int P(dx_0) Q(dy_0 | x_0)\mathcal{T}(dx_1|x_0,\gamma(y_0))Q(dy_1|x_1) \times \bigg(E[c(X_2,\gamma(Y_{[0,2]})) | x_1, y_{0},y_1 ] \bigg) \nonumber \\
&& \quad - \int P_n(dx_0) Q(dy_0|x_0) )\mathcal{T}(dx_1|x_0,\gamma(y_0))Q(dy_1|x_1) \times \bigg(E[c(X_2,\gamma(Y_{[0,2]})) | x_1, y_{0},y_1 ] \bigg)\bigg|\nonumber
\end{eqnarray}

In the last equality, we used the fact that conditioned on all observations and most recent state variable, we can take out the conditioning on the earlier state variables using the Markov properties of the system. This follows from:
\begin{align*}
&P^{\gamma}(dx_2, dy_0, dy_1, dy_2 | x_0,x_1,y_0,y_1)  \\
&= P(dy_2|x_2) P^{\gamma}(dx_2, dy_0, dy_1 | x_0,x_1,y_0,y_1) \\
&=  P(dy_2|x_2) P(dx_2 | x_1,\gamma(y_0,y_1)) P(dy_0, dy_1 | y_0, y_1) \\
&=  P(dy_2|x_2) P^{\gamma}(dx_2 | x_1,y_0,y_1) P(dy_0, dy_1 | y_0, y_1) \\
&= P^{\gamma}(dx_2, dy_0, dy_1, dy_2 | x_1,y_0,y_1),
\end{align*}
 where $P^{\gamma}$ denote the induced probability measure on the state and the measurement variables given a policy $\gamma \in \Gamma$.

Now, if we can show that the term $ \int Q(dy_0 | x_0)\mathcal{T}(dx_1|x_0,\gamma(y_0))Q(dy_1|x_1) \times \bigg(E[c(X_2,\gamma(Y_{[0,2]})) | x_1, y_{0},y_1 ]\bigg)$ is a continuous and bounded function of $x_0$ uniformly for all $\gamma \in \Gamma$, we can make the difference less than $\epsilon$ since $P_n \to P$ weakly. To show the continuity observe the following
\begin{align}
& \lim_{x_0' \to x_0}\sup_{\gamma \in \Gamma}\bigg|\int Q(dy_0|x_0) \mathcal{T}(dx_1|x_0,\gamma(y_0)) Q(dy_1|x_1)E[c(X_2,\gamma(Y_{[0,2]})) | y_{0}, x_1,y_1) ]    \nonumber \\
& \qquad - \int Q(dy_0|x_0') \mathcal{T}(dx_1|x_0',\gamma(y_0))Q(dy_1|x_1)E[c(X_2,\gamma(Y_{[0,2]})) |  y_0,x_1,y_1] \bigg|\nonumber \\
&\leq   \lim_{x_0' \to x_0}\sup_{\gamma \in \Gamma}\bigg|\int Q(dy_0|x_0)\int \mathcal{T}(dx_1|x_0',\gamma(y_0)) Q(dy_1|x_1) \times E[c(X_2,\gamma(Y_{[0,2]})) | x_1, y_{0},y_1 ] \nonumber \\
& \quad - \int Q(dy_0|x_0') \int \mathcal{T}(dx_1|x_0',\gamma(y_0)) Q(dy_1|x_1) \times  E[c(X_2,\gamma(Y_{[0,2]})) | x_1, y_{0},y_1 ]\bigg|\nonumber \\
& + \lim_{x_0' \to x_0}\sup_{\gamma \in \Gamma}\int Q(dy_0|x_0)\bigg|\int \mathcal{T}(dx_1|x_0,\gamma(y_0)) \int Q(dy_1|x_1)\times E[c(X_2,\gamma(Y_{[0,2]})) | x_1, y_{0},y_1]\nonumber\\
& \quad -\int \mathcal{T}(dx_1|x_0',\gamma(y_0)) \int Q(dy_1|x_1) \times E[c(X_2,\gamma(Y_{[0,2]})) | x_1, y_{0},y_1 ]\bigg|.\nonumber\\
\label{multiWeak2}
\end{align}
In (\ref{multiWeak2}), the first term goes to 0 by Assumption \ref{weak:as3}(b) with the following argument,
\begin{align*}
&\lim_{x_0'\to x_0}\sup_{\gamma \in \Gamma}\bigg|\int Q(dy_0|x_0) \mathcal{T}(dx_1|x_0',\gamma(y_0)) Q(dy_1|x_1)E[c(X_2,\gamma(Y_{[0,2]})) | y_{0}, x_1,y_1]  \nonumber \\
& \qquad-\int Q(dy_0|x_0') \mathcal{T}(dx_1|x_0',\gamma(y_0))Q(dy_1|x_1)E[c(X_2,\gamma(Y_{[0,2]})) | y_0,x_1,y_1]\bigg| \nonumber \\
&\leq\lim_{x_0'\to x_0}\|c\|_\infty\|Q(y_0 \in \cdot|x_0)-Q(y_0 \in \cdot|x_0')\|_{TV}=0.
\end{align*}

Before analyzing the second term, we claim the following: For a family of functions, $\{f(\gamma,x_k)\}$ which is uniformly bounded,  and equicontinuous over $\gamma \in \Gamma$,
\begin{align}\label{arzela_asc}
 \lim_{x_{k-1}' \to x_{k-1}}\sup_{\gamma \in \Gamma} \bigg|\int &\mathcal{T}(dx_k|x_{k-1},\gamma(y_{[0,k-1]}))f(\gamma,x_k)\nonumber\\
&- \int \mathcal{T}(dx_k|x_{k-1}',\gamma(y_{[0,k-1]}))f(\gamma,x_k)\bigg|=0.
\end{align}
To see this, observe the following: 
By the Arzela-Ascoli Theorem, for any given compact set $K$, and $\eta>0$ there is a finite family of continuous functions $\mathds{F}:=\{f_1,\dots,f_N\}$ such that for any $\gamma$, we can find an $f_i \in \mathds{F}$ with $\sup_{x_k \in K}|f(\gamma,x_k)-f_i(x_k) |\leq \eta$. Furthermore, since $\mathcal{T}(dx_k|x_{k-1},u_{k-1})$ is continuous, for any $r>0$, the set $\{\mathcal{T}(dx_k|x'_{k-1},u_{k-1}), |x'_{k-1}-x_{k-1}| \leq r, u_{k-1} \in \mathds{U}\}$ as a continuous image of a compact set is itself weakly compact, and thus tight, and hence, for every $\epsilon > 0$, there exists a compact set $K_{\epsilon}$ such that $\int_{K_{\epsilon}} \mathcal{T}(dx_k|x_{k-1}',\gamma(y_{[0,k-1]})) \geq 1- \epsilon$ for all $|x_{k-1}'-x_{k-1}|\leq r$ and $\gamma \in \Gamma$. 

Now, we fix an $\epsilon>0$, choose an $r>0$ and a compact set $K_\epsilon$ according to above discussion and fix a finite family of continuous functions $\mathds{F}:=\{f_1,\dots,f_N\}$ such that for any $\gamma$, we can find an $f_i \in \mathds{F}$ with $\sup_{x_k \in K}|f(\gamma,x_k)-f_i(x_k) |\leq \epsilon$. We also choose a $0<\delta\leq r$ such that $\sup_{\gamma \in \Gamma} |\int_{K_\epsilon} \mathcal{T}(dx_k|x_{k-1},\gamma(y_{[0,k-1]}))f_i(x_k) - \int_{K_\epsilon} \mathcal{T}(dx_k|x_{k-1}',\gamma(y_{[0,k-1]}))f_i(x_k)\big|\leq \epsilon$ for $|x_{k-1}'-x_{k-1}|\leq \delta$ and for all $f_i\in \mathds{F}$, which we can do since $\mathcal{T}(\cdot|x,u)$ is weakly continuous in $x$ uniformly for $u \in \mathds{U}$ and there are finitely many $f_i$.  With this setup, let us look at (\ref{arzela_asc}) again for $\{x_{k-1}':|x_{k-1}'-x_{k-1}|\leq \delta\}$.
\begin{align*}
&\sup_{\gamma \in \Gamma}\big|\int \mathcal{T}(dx_k|x_{k-1},\gamma(y_{[0,k-1]})f(\gamma,x_k) -\int \mathcal{T}(dx_k|x_{k-1}',\gamma(y_{[0,k-1]}))f(\gamma,x_k)\big|\\
&\leq\sup_{\gamma \in \Gamma}\big|\int_{\mathds{X}\setminus K_\epsilon} \mathcal{T}(dx_k|x_{k-1},\gamma(y_{[0,{k-1}]}))f(\gamma,x_k) -\int_{\mathds{X}\setminus K_\epsilon} \mathcal{T}(dx_k|x_{k-1}',\gamma(y_{[0,k-1]}))f(\gamma,x_k)\big|\\
&\qquad+\sup_{\gamma \in \Gamma}\big|\int_{K_\epsilon} \mathcal{T}(dx_k|x_{k-1},\gamma(y_{[0,k-1]}))f(\gamma,x_k) -\int_{K_\epsilon} \mathcal{T}(dx_k|x_{k-1}',\gamma(y_{[0,k-1]})f(\gamma,x_k)\big|\\
&\leq 2\epsilon\|c\|_\infty + \sup_{\gamma \in \Gamma}\big|\int_{K_\epsilon} \mathcal{T}(dx_k|x_{k-1},\gamma(y_{[0,k-1]}))\big(f(\gamma,x_k) -f_i(x_k)\big)\\
&\qquad\qquad+ \int_{K_\epsilon} \mathcal{T}(dx_k|x_{k-1},\gamma(y_{[0,k-1]}))f_i(x_k) - \int_{K_\epsilon} \mathcal{T}(dx_k|x_{k-1}',\gamma(y_{[0,k-1]}))f_i(x_k)\\
&\qquad\qquad + \int_{K_\epsilon} \mathcal{T}(dx_k|x_{k-1}',\gamma(y_{[0,k-1]}))\big(f_i(x_k)-f(\gamma,x_k)\big)\big|\\
&\leq 2\epsilon\|c\|_\infty + \sup_{\gamma \in \Gamma} \int_{K_\epsilon} \mathcal{T}(dx_k|x_{k-1},\gamma(y_{[0,k-1]}))f_i(x_k)\\
&\phantom{xxxxxxxxxxxxxxxxxx} - \int_{K_\epsilon} \mathcal{T}(dx_k|x_{k-1}',\gamma(y_{[0,k-1]}))f_i(x_k)\big| + 2\epsilon\\
&\leq2\epsilon\|c\|_\infty + 3\epsilon
\end{align*}
where $f_i(x_k)$ is chosen according to the discussion above such that $f_i$ is $\epsilon$ close to $f(\gamma,x_k)$.
As $\epsilon$ is arbitrary, (\ref{arzela_asc}) holds true.

Now we return to (\ref{multiWeak2}) again. To show that the second term also goes to 0, we use (\ref{arzela_asc}) with the claim that 
\begin{align*}
f(\gamma,x_1):= \int Q(dy_1|x_1)\times E[c(X_2,\gamma(Y_{[0,2]})) | x_1, y_{0},y_1]
\end{align*}
is a continuous and bounded function of $x_1$ uniformly for all $\gamma \in \Gamma$. To see this, we write
\begin{align*}
&\lim_{x_1' \to x_1}\sup_{\gamma \in \Gamma}\bigg( \int Q(dy_1|x_1)\times E[c(X_2,\gamma(Y_{[0,2]})) | x_1, y_{0},y_1]  \\
& \qquad \quad \quad \quad -  \int Q(dy_1|x_1')\times E[c(X_2,\gamma(Y_{[0,2]})) | x_1', y_{0},y_1]\bigg)\\
&\leq \lim_{x_1' \to x_1}\sup_{\gamma \in \Gamma}\bigg| \int Q(dy_1|x_1)\times E[c(X_2,\gamma(Y_{[0,2]})) | x_1, y_{0},y_1] \\
& \qquad \quad \quad \quad -  \int Q(dy_1|x_1')\times E[c(X_2,\gamma(Y_{[0,2]})) | x_1, y_{0},y_1]\bigg|\\
&+\lim_{x_1' \to x_1}\sup_{\gamma \in \Gamma} \int Q(dy_1|x_1)\times\bigg| E[c(X_2,\gamma(Y_{[0,2]})) | x_1, y_{0},y_1]- E[c(X_2,\gamma(Y_{[0,2]})) | x_1', y_{0},y_1]\bigg|\\
&=0.
\end{align*}

Above, for the first term we used Assumption \ref{weak:as3}(b), and for the second term we have the following
\begin{align}\label{arzela_step}
&\lim_{x_1' \to x_1}\sup_{\gamma \in \Gamma}\big| E[c(X_2,\gamma(Y_{[0,2]})) | x_1,y_0,y_1 ]- E[c(X_2,\gamma(Y_{[0,2]})) | x_1',y_0,y_1 ]\big|\nonumber\\
&=\lim_{x_1' \to x_1}\sup_{\gamma \in \Gamma}\big|\int \mathcal{T}(dx_2|x_1,\gamma(y_{[0,1]}))Q(dy_2|x_2)c(x_2,\gamma(y_{[0,2]}))\nonumber \\
& \qquad\qquad \qquad -\int \mathcal{T}(dx_2|x_1',\gamma(y_{[0,1]}))Q(dy_2|x_2)c(x_2,\gamma(y_{[0,2]}))\big|.
\end{align}

For the last term, we used (\ref{arzela_asc}) with $f(\gamma,x_2)=\int Q(dy_2|x_2)c(x_2,\gamma(y_{[0,2]}))$. To see that   $f(\gamma,x_2)$ is a uniformly bounded, equicontinuous family of functions over $\gamma \in \Gamma$, observe the following;
\begin{align*}
&\lim_{x_2' \to x_2}\sup_{\gamma \in \Gamma}\big|\int Q(dy_2|x_2')c(x_2',\gamma(y_{[0,2]}))-\int Q(dy_2|x_2)c(x_2,\gamma(y_{[0,2]}))\big|\\
&\leq \lim_{x_2' \to x_2}\sup_{\gamma \in \Gamma}\big| \int Q(dy_2|x_2')c(x_2',\gamma(y_{[0,2]})) - \int Q(dy_2|x_2)c(x_2',\gamma(y_{[0,2]}))\big| \\
&\quad \quad  +\lim_{x_2' \to x_2}\sup_{\gamma \in \Gamma}\big| \int Q(dy_2|x_2)\big(c(x_2',\gamma(y_{[0,2]})) - c(x_2,\gamma(y_{[0,2]}))\big)\big|\\
&=0.
\end{align*} 
For the first term we used that the channel $Q$ is continuous in total variation and for the second term we used that $c(x_2,\gamma(y_{[0,2]}))$ is continuous in $x_2$ uniformly for all $\gamma \in \Gamma$, because for every $\gamma$, $\gamma(y_{[0,2]})$ corresponds to a control action $u \in \mathds{U}$ and $\mathds{U}$ is compact. 

Thus, we can find an $N$ such that for all $n>N$, 
\[\sup_{\gamma \in \Gamma}\bigg|E_{P}\big[c(X_2,\gamma(Y_{[0,2]}))\big] - E_{P_n}\big[c(X_2,\gamma(Y_{[0,2]})) \big] \bigg| <\epsilon.\]
Now, we generalize the argument for any time stage $k\geq0$,
\begin{align*}
& \sup_{\gamma \in \Gamma}\bigg|E_{P}\big[c(X_k,\gamma(Y_{[0,k]}))\big] - E_{P_n}\big[c(X_k,\gamma(Y_{[0,k]})) \big] \bigg| \nonumber \\
& = \sup_{\gamma \in \Gamma}\bigg| \int P(dx_0) Q(dy_0 | x_0)\mathcal{T}(dx_1|x_0,\gamma(y_0))Q(dy_1|x_1)\dots \mathcal{T}(dx_{k-1}|x_{k-2},\gamma(y_{[0,k-2]}))\\
&\phantom{xxxxxxxx}Q(dy_{k-1}|x_{k-1}) \times \bigg(E[c(X_k,\gamma(Y_{[0,k]})) | x_0,x_1, \dots x_{k-1}, y_{0},\dots,y_{k-1} ] \bigg) \nonumber \\
& \quad-\int P_n(dx_0) Q(dy_0|x_0) )\mathcal{T}(dx_1|x_0,\gamma(y_0))\dots\mathcal{T}(dx_{k-1}|x_{k-2},\gamma(y_{[0,k-2]}))Q(dy_{k-1}|x_{k-1}) \nonumber\\
&\quad\quad\quad \times \bigg(E[c(X_k,\gamma(Y_{[0,k]})) | x_0,x_1, \dots x_{k-1}, y_{0},\dots,y_{k-1} ] \bigg)\bigg| \nonumber \\
&= \sup_{\gamma \in \Gamma}\bigg| \int P(dx_0) Q(dy_0 | x_0)\mathcal{T}(dx_1|x_0,\gamma(y_0))Q(dy_1|x_1)\dots \mathcal{T}(dx_{k-1}|x_{k-2},\gamma(y_{[0,k-2]}))\\
&\phantom{xxxxxxxxx}Q(dy_{k-1}|x_{k-1}) \times \bigg(E[c(X_k,\gamma(Y_{[0,k]})) | x_{k-1}, y_{0},\dots,y_{k-1} ] \bigg) \nonumber \\
& \quad-\int P_n(dx_0) Q(dy_0|x_0) )\mathcal{T}(dx_1|x_0,\gamma(y_0))\dots\mathcal{T}(dx_{k-1}|x_{k-2},\gamma(y_{[0,k-2]}))Q(dy_{k-1}|x_{k-1}) \nonumber\\
&\quad\quad\quad \times \bigg(E[c(X_k,\gamma(Y_{[0,k]})) | x_{k-1}, y_{0},\dots,y_{k-1} ] \bigg)\bigg| 
\end{align*}
the last equality follows from the Markov property of the system as we do not depend on the past state variables given all the past observations and most recent state variable.
We again need to show the continuity in $x_0$.
\begin{eqnarray}
&&\lim_{x_0' \to x_0} \sup_{\gamma \in \Gamma}\bigg[\int Q(dy_0|x_0)\dots \mathcal{T}(dx_{k-1}|x_{k-2},\gamma(y_{[0,k-2]})) E[c(X_k,\gamma(Y_{[0,k]})) | x_{k-1},y_{0}\dots y_{k-1} ]    \nonumber \\
&& \quad\quad\quad\quad - \int Q(dy_0|x_0')\dots \mathcal{T}(dx_{k-1}|x_{k-2},\gamma(y_{[0,k-2]})E[c(X_k,\gamma(Y_{[0,k]})) |x_{k-1}, y_{0},\dots,y_{k-1} ] \bigg]\nonumber \\
&& \leq \lim_{x_0' \to x_0}\sup_{\gamma \in \Gamma} \bigg|\int Q(dy_0|x_0)\int \mathcal{T}(dx_1|x_0',\gamma(y_0)) \dots T(dx_{k-1}|x_{k-2},\gamma(y_{[0,k-2]})) Q(dy_{k-1}|x_{k-1})\nonumber \\
&& \qquad \qquad \qquad\qquad\qquad  \times E[c(X_k,\gamma(Y_{[0,k]})) | x_{k-1}, y_{0},\dots,y_{k-1} ] \nonumber \\
&&\hspace{1.7cm} - \int Q(dy_0|x_0') \int \mathcal{T}(dx_1|x_0',\gamma(y_0)) \dots \mathcal{T}(dx_{k-1}|x_{k-2},\gamma(y_{[0,k-2]})) Q(dy_{k-1}|x_{k-1})\nonumber \\
&& \qquad \qquad \qquad\qquad\qquad  \times  E[c(X_k,\gamma(Y_{[0,k]})) | x_{k-1}, y_{0},\dots,y_{k-1} ]\bigg|\nonumber \\
&&\quad+ \lim_{x_0' \to x_0}\sup_{\gamma \in \Gamma} \int Q(dy_0|x_0)\bigg|\int \mathcal{T}(dx_1|x_0,\gamma(y_0)) \dots   \mathcal{T}(dx_{k-1}|x_{k-2},\gamma(y_{[0,k-2]})) \nonumber\\
&&\hspace{5cm}Q(dy_{k-1}|x_{k-1}) \times E[c(X_k,\gamma(Y_{[0,k]})) | x_{k-1}, y_{0},\dots,y_{k-1} ]\nonumber\\
&&\hspace{4cm}-\int \mathcal{T}(dx_1|x_0',\gamma(y_0)) \dots \mathcal{T}(dx_{k-1}|x_{k-2},\gamma(y_{[0,k-2]}))\nonumber\\ 
&&\hspace{5cm}Q(dy_{k-1}|x_{k-1}) \times E[c(X_k,\gamma(Y_{[0,k]})) | x_{k-1}, y_{0},\dots,y_{k-1} ]\bigg|\nonumber\\
&&=0. \label{multiWeak3}
\end{eqnarray}
Similar to the special case where $k=2$, at (\ref{multiWeak3}), for the first term we used Assumption \ref{weak:as3}(b). 

For the second term, we used the fact that the term,
\begin{align*}
\int Q(dy_1|x_1)\dots \mathcal{T}(dx_{k-1}|x_{k-2},\gamma(y_{[0,k-2]})) Q(dy_{k-1}|x_{k-1}) E[c(X_k,\gamma(Y_{[0,k]})) |x_{k-1}, y_{0},\dots,y_{k-1} ]
\end{align*}
is a continuous and bounded function of $x_1$ uniformly for all $\gamma \in \Gamma$. For this, it is necessary to show that,
\begin{align*}
\int Q(dy_2|x_2)\dots \mathcal{T}(dx_{k-1}|x_{k-2},\gamma(y_{[0,k-2]})) Q(dy_{k-1}|x_{k-1}) E[c(X_k,\gamma(Y_{[0,k]})) | x_{k-1}, y_{0},\dots,y_{k-1}]
\end{align*}
is a continuous function of $x_2$ uniformly over $\Gamma$. Then, inductively, the last step is to show that
\begin{align*}
 E[c(X_k,\gamma(Y_{[0,k]})) | x_{k-1}, y_{0},\dots,y_{k-1}]=\int c(x_k,\gamma(y_{[0,k]})) \mathcal{T}(dx_k|x_{k-1},\gamma(y_{[0,k-1]}))Q(dy_k|x_k)
\end{align*}
is a continuous function of $x_{k-1}$ uniformly over $\Gamma$, which holds since $ \mathcal{T}(dx_k|x_{k-1},\gamma(y_{[0,k-1]}))$ is weakly continuous in $x_{k-1}$ and the continuity is uniform on $\Gamma$ since for every $\gamma \in \Gamma$, $\gamma(y_{[0,k-1]})$ corresponds to a control action from $\mathds{U}$ and $\mathds{U}$ is compact. The fact that $Q(dy_k|x_k)$ is continuous in total variation in $x_k$ and $c(x_k,\gamma(y_{[0,k]}))$ is continuous in $x_k$ uniformly over $\Gamma$ completes the proof.
\end{proof}

  \subsection{Proof of Lemma~\ref{lemma:lusin}}
  The proof is adapted from the proof of Theorem 3.2 in \cite{yuksel12:siam}.
  \begin{proof}
    Let $\mu_\mathds{Y}$ satisfy $\mu_\mathds{Y}(A)=\mu(\mathds{X}\times A)$ for $A\in\mathcal{B}(\mathds{Y})$. By Lusin's theorem (\cite[Theorem 2.24]{Rud87}), for all $\gamma\in\Gamma$ and $\varepsilon>0$, there is a continuous function $f:\mathds{U}\to\mathds{Y}$ such that
    \begin{align*}
      \mu_\mathds{Y}\{y:f(y)\neq \gamma(y)\} < \varepsilon.
    \end{align*}
    For convenience of notation we let $B=\{y:f(y)\neq \gamma(y)\}$. Observe the following
    \begin{align*}
      \int \left|c(x,\gamma(y))-c(x,f(y))\right|\mu(\dd x, \dd y) &= \int_{\mathds{X}\times B} \left|c(x,\gamma(y))-c(x,f(y))\right|\mu(\dd x, \dd y)\\ &< \varepsilon \cdot \|c\|_\infty
    \end{align*}
    where $\|c\|_\infty$ denotes the supremum norm of $c$, which is finite by assumption. We have that 
    \begin{align}\label{eq:lusinbound}
      \int c(x,f(y)) \mu (\dd x, \dd y) < \int c(x,g(y)) \mu (\dd x, \dd y) + \varepsilon \cdot \|c\|_\infty.
    \end{align}
    Define
    \begin{align*}
      j(\mu,\mathcal{C})=\inf_{\gamma\in\mathcal{C}}\int c(x,\gamma(y))\mu(\dd x, \dd y),\text{ and }j(\mu,\Gamma)=\inf_{\gamma\in\Gamma}\int c(x,\gamma(y))\mu(\dd x, \dd y).
    \end{align*}
    Note that $j(\mu,\mathcal{C})\geq j(\mu,\Gamma)$ as $\mathcal{C}\subset \Gamma$. From \eqref{eq:lusinbound} we have $j(\mu,\mathcal{C})\leq j(\mu,\Gamma) + \varepsilon \cdot \|c\|_\infty$, which gives $j(\mu,\mathcal{C})\leq j(\mu,\Gamma)$ as $\varepsilon$ was arbitrary. Hence $j(\mu,\mathcal{C})= j(\mu,\Gamma)$.
  \end{proof}

  \section{Additional Proofs}\label{app}
  \subsection{Calculations for Proof of Theorem~\ref{ex:setwise}}\label{calculationsSetwise}

In the following, $p$ will denote the density of the corresponding probability measure. 

  \subsubsection{Computing $\gamma_P^*$}
  \begin{align*}
    \gamma_P^* &= E[X|Y=y]= \int_0^1 x p(x|y)\dd x= \int_0^1 x \frac{p(y|x)}{\int_0^1 p(y|x)p(x)\dd x}\dd x
  \end{align*}
  by Bayes' Rule. We now compute the term in the denominator,
  \begin{align*}
    \int_0^1 p(y|x)\cancelto{1}{p(x)}\dd x = \int_0^1 \left(\frac{1}{2} + \frac{1}{2}\delta_x(y) \right) \dd x=1.
  \end{align*}
  This gives
  \begin{align*}
    \gamma_P^* &=\int_0^1 x p(y|x)\dd x=\int_0^1 x \left(\frac{1}{2} + \frac{1}{2}\delta_x(y) \right)\dd x=\int_0^1 \frac{1}{2}x \dd x + \frac{1}{2}y\\
    &=\frac{1}{4}+\frac{1}{2}y = \frac{1}{2}\left(\frac{1}{2}+y\right)
  \end{align*}
  as desired.
  \subsubsection{Computing $J^*(P,Q)$}
  \begin{align*}
    J^*(P,Q)&=E[c(x,\gamma_P^*(y))]=E[(x-\gamma_P^*(y))^2]\\
    &=\int_0^1 \int_0^1 \left(x-\frac{1}{4}-\frac{1}{2}y\right)^2 Q(\dd y|x)P(\dd x).
  \end{align*}
  We now compute the inner integral for some fixed $x$,
  \begin{align*}
    \int_0^1 &\left(x-\frac{1}{4}-\frac{1}{2}y\right)^2 Q(\dd y|x)\\ &=\frac{1}{2} \int_0^1 \left(x-\frac{1}{4}-\frac{1}{2}y\right)^2 \dd y + \frac{1}{2} \int_0^1 \left(x-\frac{1}{4}-\frac{1}{2}y\right)^2 \delta_x(y) \dd y\\
    &=\frac{1}{2}\left(\frac{13}{48}-x+x^2\right)+\frac{1}{2}\left(\frac{1}{16}-\frac{x}{4}+\frac{x^2}{4}\right)\\
    &=\frac{1}{6}-\frac{5x}{8}+\frac{5x^2}{8}.
  \end{align*}
  This gives
  \begin{align*}
    J^*(P,Q)=\int_0^1 \left(\frac{1}{6}-\frac{5x}{8}+\frac{5x^2}{8}\right) \dd x=\frac{1}{16}.
  \end{align*}
  \subsubsection{Computing $\gamma_{P_n}^*$}\label{app:gammapn}
  Let $n\in\N$. If $y\in R$ then we have,
  \begin{align*}
    p(y)&=\int_0^1 p(y|x)f_n(x) \dd x=\int_L \frac{1}{2}(2) \dd x =\int_L \dd x\\
    &=\lambda(L) = \frac{1}{2},
  \end{align*}
  where $\lambda$ is the Lebesgue measure. We now compute $E[X|Y=y]$,
  \begin{align}
    E[X|Y=y]&=\int_0^1 x p(x|y)\dd x\nonumber \\
    &= \int_0^1 x \frac{p(y|x)f_n(x)}{\int_0^1 p(y|x)P(x)\dd x}\dd x\nonumber\\
    &= \int_L x \frac{1/2\cdot2}{1/2} \dd x\nonumber\\
    &= \sum_{k=1}^n \left.x^2\right|_{\frac{2k-2}{2n}}^{\frac{2k-1}{2n}}= \sum_{k=1}^n \left( \left(\frac{2k-1}{2n}\right)^2-\left(\frac{2k-2}{2n}\right)^2 \right)\nonumber\\
    &= \sum_{k=1}^n \frac{4k-3}{4n^2} = \frac{1}{n^2} \sum_{k=1}^n k-\frac{1}{n^2} \sum_{k=1}^n \frac{3}{4}\nonumber\\
    &= \frac{1}{n^2} \cdot \frac{n(n+1)}{2}-\frac{1}{n^2}\cdot\frac{3n}{4}= \frac{1}{2}-\frac{1}{4n}\label{expectation1},
  \end{align}
  where we have used Bayes' Rule to establish the second equality. Now we consider the case where $y\in L$. We calculate $p(y)$,
  \begin{align*}
    p(y)&=\int_0^1 p(y|x)f_n(x) \dd x=\int_L \frac{1}{2}(2) \dd x + \frac{1}{2}\cdot 2 =\int_L \dd x + 1\\
    &=\lambda(L) + 1 = \frac{3}{2}.
  \end{align*}
  This gives,
  \begin{align*}
     E[X|Y=y]&=\int_0^1 x p(x|y)\dd x\\
    &= \int_0^1 x \frac{p(y|x)f_n(x)}{\int_0^1 p(y|x)p(x)\dd x}\dd x\\
    &= \int_L x \frac{1/2\cdot2}{3/2} \dd x + \int_L x \frac{1/2\cdot2}{3/2} \delta_x (y) \dd x \\
    &= \frac{1}{3}\cdot\left(\frac{1}{2}-\frac{1}{4n}\right)+\frac{2}{3}y,
  \end{align*}
  where we use \eqref{expectation1} to give the last equality. We see that the optimal policy is
  \begin{align*}
      \gamma_{P_n}^*(y)=E[X|Y=y]=
      \begin{cases}
        \frac{1}{2}-\frac{1}{4n} & \text{if } y\in R\\
        \frac{1}{3}\cdot\left(\frac{1}{2}-\frac{1}{4n}\right)+\frac{2}{3}y & \text{if } y\in L
      \end{cases}.
  \end{align*}
  \subsubsection{Computing $J^*(P_n,Q)$}\label{app:JstarPn}
  Let $n\in\N$.
  \begin{align*}
    J^*(P_n,Q)&=E[c(x,\gamma_{P_n}^*(y))]\\
    &=E[(x-\gamma_{P_n}^*(y))^2]\\
    &= \int_0^1 \left( \int_0^1(x-\gamma_{P_n}^*(y))^2 Q(\dd y|x)\right) P_n(\dd x)\\
    &=2\int_L \left( \underbrace{\int_L(x-\gamma_{P_n}^*(y))^2 Q(\dd y|x)}_{\text{(i)}} +  \underbrace{\int_R(x-\gamma_{P_n}^*(y))^2 Q(\dd y|x)}_{\text{(ii)}}\right) \dd x
  \end{align*}
  We now compute term (i) for $x\in L$:
  \begin{align*}
    \int_L(x-\gamma_{P_n}^*(y))^2 Q(\dd y|x)&=\int_L\left(x-\frac{1}{3}\left(\frac{1}{2}-\frac{1}{4n}\right)-\frac{2}{3}y\right)^2 Q(\dd y|x)\\
    &=\frac{1}{2} \left(x-\frac{1}{3}\left(\frac{1}{2}-\frac{1}{4n}\right)-\frac{2}{3}x\right)^2+\dots\\&\quad+\int_L\frac{1}{2}\left(x-\frac{1}{3}\left(\frac{1}{2}-\frac{1}{4n}\right)-\frac{2}{3}y\right)^2 \dd y\\
    &=\frac{1}{72}+\frac{1}{288 n^2}-\frac{1}{72n}+\left(-\frac{1}{18}+\frac{1}{36 n}\right) x+\frac{x^2}{18}+\dots\\&\quad+\int_L \left(\frac{1}{72}+\frac{1}{288 n^2}-\frac{1}{72 n}+\left(-\frac{1}{6}+\frac{1}{12 n}\right)x+\dots\right.\\&\quad\left.+\frac{x^2}{2}+\left(\frac{1}{9}-\frac{1}{18 n}-\frac{2}{3}x\right) y+\frac{2y^2}{9} \right) \dd y\\
    &=\frac{1}{72}+\frac{1}{288 n^2}-\frac{1}{72n}+\left(-\frac{1}{18}+\frac{1}{36 n}\right) x+\frac{x^2}{18}+\dots\\&\quad+ \left(\frac{1}{72}+\frac{1}{288 n^2}-\frac{1}{72 n}+\left(-\frac{1}{6}+\frac{1}{12 n}\right)x+\frac{x^2}{2}\right)\int_L \dd y +\dots\\&\quad+\left(\frac{1}{9}-\frac{1}{18 n}-\frac{2}{3}x\right)\int_L y \dd y + \frac{2}{9} \int_L y^2 \dd y 
  \end{align*}
  For ease we compute the integrals in the above step separately. Some are familiar from Section~\ref{app:gammapn}.
  \begin{align*}
    \int_L \dd y &= \lambda(L) = \frac{1}{2},\\
    \int_L y \dd y &= \frac{1}{2}\sum_{k=1}^n \left.y^2\right|_{\frac{2k-2}{2n}}^{\frac{2k-1}{2n}}\\
    &= \frac{1}{2}\left(\frac{1}{2}-\frac{1}{4n}\right)\text{ (by \eqref{expectation1})},\\
    \int_L y^2 \dd y &=\frac{1}{3}\sum_{k=1}^n \left.y^3\right|_{\frac{2k-2}{2n}}^{\frac{2k-1}{2n}}\\
    &= \frac{1}{3}\left(\sum_{k=1}^n \left( \left(\frac{2k-1}{2n}\right)^3-\left(\frac{2k-2}{2n}\right)^3 \right)\right)\\
    &= \frac{1}{3}\left( \sum_{k=1}^n \frac{7-18k+12k^2}{8n^3}\right)\\
    &= \frac{1}{3}\left( \frac{7}{8n^3} \sum_{k=1}^n 1 -\frac{9}{4n^3}\sum_{k=1}^n k+\frac{3}{2n^3}\sum_{k=1}^n k^2 \right)\\
    &=  \frac{7}{24n^3} (n) -\frac{3}{4n^3}\left(\frac{n^2+n}{2}\right)+\frac{1}{2n^3}\left(\frac{2n^3+3n^2+n}{6}\right)\\
    &=\frac{1}{6}-\frac{1}{8n}.
  \end{align*}
  These combine with the above to give (after some simplification),
  \begin{align*}
    \int_L(x-\gamma_{P_n}^*(y))^2 Q(\dd y|x)=\frac{37}{432}+\frac{7}{576 n^2}-\frac{11}{144n}+\left(-\frac{11}{36}+\frac{11}{72 n}\right) x+\frac{11}{36}x^2.
  \end{align*}
  Next we compute term (ii) for $x\in L$:
  \begin{align*}
    \int_R(x-\gamma_{P_n}^*(y))^2 Q(\dd y|x)&=\frac{1}{2}\int_R \left(x-\frac{1}{2}+\frac{1}{4n}\right)^2 \dd y=\frac{1}{4}\left(x-\frac{1}{2}+\frac{1}{4n}\right)^2\\
    &=\frac{1}{16}+\frac{1}{64 n^2}-\frac{1}{16 n}+\left(-\frac{1}{4}+\frac{1}{8 n}\right) x+\frac{x^2}{4}.
  \end{align*}
  For $x\in L$ we compute (i)$+$(ii),
  \begin{align*}
    \text{(i)}+\text{(ii)}&=\frac{4}{27}+\frac{1}{36 n^2}-\frac{5}{36 n}+\left(-\frac{5}{9}+\frac{5}{18 n}\right)x+\frac{5}{9}x^2.
  \end{align*}
  This gives,
  \begin{align*}
    J(P_n,Q)&=\int_0^1 \left(c(x,\gamma_{P_n}^*(y)\right)^2 P_n(\dd x)\\
    &=2\int_L \left( \frac{4}{27}+\frac{1}{36 n^2}-\frac{5}{36 n}+\left(-\frac{5}{9}+\frac{5}{18 n}\right)x+\frac{5}{9}x^2\right) \dd x\\
    &=2\left( \frac{4}{27}+\frac{1}{36 n^2}-\frac{5}{36 n}\right)\int_L  \dd x+2\left(-\frac{5}{9}+\frac{5}{18 n}\right)\int_L x \dd x+2\left(\frac{5}{9}\right)\int_L x^2 \dd x\\
    &=\frac{1}{18}-\frac{1}{24n^2}.
  \end{align*}

  \subsection{Calculations for Proof of Proposition~\ref{prop:SWextension}}\label{app:SWextension}
  Let $n\in\N$. From Section~\ref{app:gammapn} we know that 
  \begin{align*}
    \gamma_{P_n}^*(y)=E[X|Y=y]=
      \begin{cases}
        \frac{1}{2}-\frac{1}{4n} & \text{if } y\in R\\
        \frac{1}{3}\cdot\left(\frac{1}{2}-\frac{1}{4n}\right)+\frac{2}{3}y & \text{if } y\in L
      \end{cases}.
  \end{align*}
  With $P\sim U([0,1])$ we have,
  \begin{align*}
    &J(P,Q,\gamma_{P_n}^*)=E[c(x,\gamma_{P_n}^*(y))]=E[(x-\gamma_{P_n}^*(y))^2]\\
    &= \int_0^1 \left( \int_0^1(x-\gamma_{P_n}^*(y))^2 Q(\dd y|x)\right) P(\dd x)\\
    &=\int_L \left( \int_0^1(x-\gamma_{P_n}^*(y))^2 Q(\dd y|x)\right) \dd x+\int_R\left( \int_0^1(x-\gamma_{P_n}^*(y))^2 Q(\dd y|x)\right) \dd x\\
    &=\frac{1}{2}J^*(P_n,Q)+\int_R \left(\int_L(x-\gamma_{P_n}^*(y))^2 Q(\dd y|x) +  \int_R(x-\gamma_{P_n}^*(y))^2 Q(\dd y|x)\right) \dd x\\
    &=\frac{1}{36}-\frac{1}{48n^2}+\int_R \left( \underbrace{\int_L(x-\gamma_{P_n}^*(y))^2 Q(\dd y|x)}_{\text{(iii)}} +  \underbrace{\int_R(x-\gamma_{P_n}^*(y))^2 Q(\dd y|x)}_{\text{(iv)}}\right) \dd x,
  \end{align*}
  where we have used the computation from Section~\ref{app:JstarPn} for $J^*(P_n,Q)$. For $x\in R$ we compute (iii):
  \begin{align*}
    \int_L(x-\gamma_{P_n}^*(y))^2 Q(\dd y|x)&=\int_L\left(x-\frac{1}{3}\left(\frac{1}{2}-\frac{1}{4n}\right)-\frac{2}{3}y\right)^2 Q(\dd y|x)\\
    &=\int_L\frac{1}{2}\left(x-\frac{1}{3}\left(\frac{1}{2}-\frac{1}{4n}\right)-\frac{2}{3}y\right)^2 \dd y\\
    &=\left(\frac{1}{72}+\frac{1}{288 n^2}-\frac{1}{72 n}+\left(-\frac{1}{6}+\frac{1}{12 n}\right)x+\frac{x^2}{2}\right)\int_L \dd y +\dots\\&\quad+\left(\frac{1}{9}-\frac{1}{18 n}-\frac{2}{3}x\right)\int_L y \dd y + \frac{2}{9} \int_L y^2 \dd y,
  \end{align*}
  where the last line follows from the computation of (i) in Section~\ref{app:JstarPn}. After some simplification we arrive at:
  \begin{align*}
    \int_L(x-\gamma_{P_n}^*(y))^2 Q(\dd y|x)&=\frac{31}{432}+\frac{5}{576n^2}-\frac{1}{16n}+\left(-\frac{1}{4}+\frac{1}{8n}\right)x+\frac{x^2}{4}.
  \end{align*}
  For $x\in R$ we compute term (iv):
  \begin{align*}
    \int_R(x-\gamma_{P_n}^*(y))^2 Q(\dd y|x)&=\int_R \left(x-\frac{1}{2}+\frac{1}{4n}\right)^2 Q(\dd y|x)\\
    &=\frac{1}{2}\left(x-\frac{1}{2}+\frac{1}{4n}\right)^2 + \frac{1}{2}\int_R\left(x-\frac{1}{2}+\frac{1}{4n}\right)^2 \dd y\\
    &=\frac{3}{4}\left(x-\frac{1}{2}+\frac{1}{4n}\right)^2\\
    &=\frac{3}{16}+\frac{3}{64 n^2}-\frac{3}{16 n}+\left(-\frac{3}{4}+\frac{3}{8 n}\right) x+\frac{3}{4}x^2.
  \end{align*}
  For $x\in R$ we have,
  \begin{align*}
    \text{(iii)}+\text{(iv)}&=\frac{7}{27}+\frac{1}{18 n^2}-\frac{1}{4 n}+\left(-1+\frac{1}{2 n}\right) x+x^2.
  \end{align*}
  This gives
  \begin{align*}
   & \int_R\left( \int_0^1(x-\gamma_{P_n}^*(y))^2 Q(\dd y|x)\right) \dd x=\int_R\left( \frac{7}{27}+\frac{1}{18 n^2}-\frac{1}{4 n}+\left(-1+\frac{1}{2 n}\right) x+x^2\right) \dd x\\
    &=\left( \frac{7}{27}+\frac{1}{18 n^2}-\frac{1}{4 n}\right)\int_R \dd x+\left(-1+\frac{1}{2 n}\right) \int_R x \dd x+ \int_R x^2 \dd x.
  \end{align*}
  We compute the integrals in the previous step separately:
  \begin{align*}
    \int_R \dd x&= \lambda(R)=\frac{1}{2},\\
    \int_R x \dd x&= \int_0^1 x \dd x - \int_L x \dd x=\frac{1}{2}-\left(\frac{1}{4}-\frac{1}{8n}\right)=\frac{1}{4}+\frac{1}{8n},\\
    \int_R x^2 \dd x&= \int_0^1 x^2 \dd x - \int_L x^2 \dd x=\frac{1}{3}-\left(\frac{1}{6}-\frac{1}{8n}\right)=\frac{1}{6}+\frac{1}{8n}.
  \end{align*}
  With the above, this gives (after some simplification):
  \begin{align*}
    \int_R\left( \int_0^1(x-\gamma_{P_n}^*(y))^2 Q(\dd y|x)\right) \dd x &=\frac{5}{108}+\frac{13}{144n^2}.
  \end{align*}
  Therefore, $J(P,Q,\gamma_{P_n}^*)=\frac{2}{27}+\frac{5}{72n^2}.$

\bibliographystyle{plain}

\begin{thebibliography}{10}
\bibitem{basbern} {\sc T. Ba\c{s}ar and P. Bernhard}, {\em H-infinity Optimal Control and Related Minimax Design Problems: A Dynamic Game Approach},
Birkh\"auser, Boston, MA, 1995.

\bibitem{YukselBaker}{\sc G. Baker and S. Y\"uksel}, {\em Continuity and Robustness to Incorrect Priors in Estimation and Control }, Proceedings of the IEEE Int. Symp. on Information Theory (ISIT), 2016, pp. ~1999-2013.
  
\bibitem{benavoli2011robust} {\sc A. Benavoli and L. Chisc}, {\em Robust Stochastic Control Based on Imprecise Probabilities}, IFAC Proceedings Volumes, 44(1), 2011, pp. ~ 4606-4613.
  
\bibitem{bertsekas78} {\sc D. P. Bertsekas and S. Shreve}, {\em Stochastic Optimal Control: The Discrete Time Case}, Academic Press, New York, 1978.

\bibitem{Billingsley} {\sc P. Billingsley}, {\em Convergence of Probability Measures}, Wiley, New York, 1968.

\bibitem{billingsley1967uniformity} {\sc P. Billingsley and F. Tops{\O}e}, {\em Uniformity in weak convergence}, Probability Theory and Related Fields, 7(1), 1967, pp. ~1--16.
 

\bibitem{Borkar2} {\sc V. S. Borkar}, {\em Convex Analytic Methods in Markov Decision Processes}, in {H}andbook of {M}arkov Decision Processes, E. A. Feinberg, A. Shwartz (Eds.), Kluwer, Boston, MA, 2001, pp. ~347--375.

\bibitem{budhiraja1997exponential} {\sc A. Budhiraja and D. Ocone}, {\em Exponential stability of discrete-time filters for bounded observation noise}, Systems Control Lett., 30(4), 1997, pp.~185--193.

\bibitem{budhiraja1999exponential} {\sc A. Budhiraja and D. Ocone}, {\em Exponential stability in discrete-time filtering for non-ergodic signals}, Stochastic Process. Appl., 82(2), 1999, pp.~245--257.
  

\bibitem{CharalambousRezaei07} {\sc C. D. Charalambous and F. Rezaei}, {\em Stochastic uncertain systems subject to relative entropy constraints: Induced norms and monotonicity properties of minimax games}, IEEE Trans. Automat. Control, 52(4), 2007, pp. ~647--663.

\bibitem{chigansky2009intrinsic} {\sc P. Chigansky, and R. Liptser, and R. van Handel}, {\em Intrinsic methods in filter stability}, 
  title={Intrinsic methods in filter stability}, Handbook of Nonlinear Filtering, 2009.
  
\bibitem{Devroye85} {\sc L. Devroye and L. Gy\"orfi}, {\em Non-parametric Density Estimation: The $L_1$ View}, Wiley, New York, 1985.

\bibitem{dudley1999uniform} {\sc R. M. Dudley}, {\em Uniform Central Limit Theorems}, Cambridge Univ. Press, Cambridge,  23, 1999.

\bibitem{Dudley02} {\sc R. M. Dudley}, {\em Real Analysis and Probability}, 2nd edition, Cambridge Univ. Press, Cambridge, 2002.

\bibitem{Dudley} {\sc R. M. Dudley and E. Gine and J. Zinn}, {\em Uniform and Universal Glivenko-Cantelli Classes}, J. Theoret. Probab., 4, 1991, pp.~485--510.

\bibitem{dupuis2000robust} {\sc P. Dupuis and M. R. James and I. Petersen}, {\em Robust Properties of Risk-Sensitive Control}, Math. Control Signals Systems, 13(4), 2000, pp.~318--332.
  
\bibitem{dynkin1979controlled} {\sc E. B. Dynkin and A. A. Yushkevich}, {\em Controlled Markov Processes}, Springer, 1979.
 

\bibitem{feinberg1996measurability} {\sc E. A. Feinberg}, {\em On Measurability and Representation of Strategic Measures in Markov Decision Processes}, Lecture Notes-Monograph Series, 1996, pp.~29--43.


\bibitem{FeKaZg14} {\sc E.A.~Feinberg and P.O.~Kasyanov and M.Z.~Zgurovsky}, {\em Partially Observable Total-Cost Markov Decision Process with Weakly Continuous Transition Probabilities}, Math. Oper. Res., 41(2), 2016, pp.~656--681.


\bibitem{goeva2016reconstructing} {\sc A. Goeva and H. Lam and H. Qian and B. Zhang}, {\em Reconstructing Input Models in Stochastic Simulation}, Proceed. the 2014 Winter Simulation Conf.,  IEEE Press, 2014, pp.~698--709.


\bibitem{gossner2008entropy} {\sc O. Gossner and T. Tomala}, {\em Entropy Bounds on Bayesian Learning}, J. Math. Econom., 44(1), 2008, pp.~24--32.

\bibitem{GrayInfo} {\sc R. M. Gray}, {\em Entropy and Information Theory}, Springer, New York, 1990.


\bibitem{gupta2014existence} {\sc A. Gupta, and S. Y\"uksel, and T. Ba\c{s}ar, and C. Langbort}, {\em On the Existence of Optimal Policies for a Class of Static and Sequential Dynamic Teams}, SIAM J. Control Optim., 53, 2015, pp.~1681--1712.
  
\bibitem{Ramon2008discrete} {\sc R. Van Handel}, {\em Discrete Time Nonlinear Filters with Informative Observations are Stable}, Electron. Commun. Probab. Volume 13, 2008, pp.~562--575.

\bibitem{hansen2001robust} {\sc L. P. Hansen, and T. J. Sargent}, {\em Robust Control and Model Uncertainty}, American Economic Review, 91(2), 2001, pp.~60--66.

\bibitem{Her89} {\sc O.~Hern\'andez-Lerma}, {\em Adaptive Markov Control Processes}, Springer, 1989.

\bibitem{HernandezLermaMCP} {\sc O. Hernandez-Lerma and J. B. Lasserre}, {\em Discrete-Time Markov Control Processes: Basic Optimality Criteria}, Springer, 1996.

\bibitem{hernandezlasserre1999further} {\sc O. Hern{\'a}ndez-Lerma, and J. B. Lasserre}, {\em Further Topics on Discrete-Time Markov Control Processes}, Springer, 1999.
  
\bibitem{Hernandez} {\sc O. Hern\'andez-Lerma and J. B. Lasserre}, {\em Markov Chains and Invariant Probabilities}, Birkh\"auser, Basel, 2003.


\bibitem{KellyWeakConv} {\sc C-H. Hsieh and B. R. Barmish and J. A. Gubner}, {\em Kelly Betting Can Be Too Conservative}, IEEE Conf. on Decision and Control, 2016.
 

\bibitem{jacobson1973optimal} {\sc D. Jacobson}, {\em Optimal Stochastic Linear Systems with Exponential Performance Criteria and Their Relation to Deterministic Differential Games}, IEEE Trans. Automat. Control, 18(2), 1973, pp.~124--131.

\bibitem{lam2016} {\sc H. Lam}, {\em Robust Sensitivity Analysis for Stochastic Systems}, Math. Oper. Res., 41(4), 2016, pp.~1248--1275.


\bibitem{Langen81} {\sc H.J. Langen}, {\em Convergence of Dynamic Programming Models}, Math. Oper. Res., 6(4), 1981, pp.~493--512.

\bibitem{oksendal2014forward} {\sc B. {\O}ksendal and A. Sulem}, {\em Forward--Backward Stochastic Differential Games and Stochastic Control Under Model Uncertainty}, J. Optim. Theory Appl., 161(1), 2014, pp.~22--55.
 

\bibitem{Par67} {\sc K.R.~Parthasarathy}, {\em Probability Measures on Metric Spaces}, AMS Bookstore, 1967.


\bibitem{dupuis2000kernel} {\sc I. Petersen, M. R. James and P. Dupuis}, {\em Minimax Optimal Control of Stochastic Uncertain Systems with Relative Entropy Constraints}, IEEE Trans. Automat. Control, 45(3), 2000, pp.~398--412.

\bibitem{dai1996connections} {\sc P. Dai Pra, L. Meneghini and W. J. Runggaldier}, {\em Connections Between Stochastic Control and Dynamic Games}, Math. Control Signals Systems, 9(4), 1996, pp.~303--326.
 
\bibitem{raginsky2013empirical} {\sc M. Raginsky}, {\em Empirical Processes, Typical Sequences and Coordinated Actions in Standard Borel Spaces}, IEEE Trans. Inform. Theory, 59(3), 2013, pp.~1288--1301.
  

\bibitem{Rhe74} {\sc D.~Rhenius}, {\em Incomplete Information in Markovian Decision Models}, Ann. Statist., 2, 1974, pp.~1327--1334.


\bibitem{Rud87} {\sc W. Rudin}, {\em Real and Complex Analysis}, 3rd edition, McGraw-Hill, New York, 1987.

\bibitem{saldi2014near} {\sc N. Saldi, and S. Y\"uksel, and T. Linder}, {\em Near Optimality of Quantized Policies in Stochastic Control Under Weak Continuity Conditions}, J. Math. Anal. Appl., 435(1), 2016, pp.~321--337.
  

\bibitem{Schal} {\sc M. Sch\"al}, {\em Conditions for Optimality in Dynamic Programming and for the Limit of
n-Stage Optimal Policies to Be Optimal}, Z. Wahrscheinlichkeitsth, 32(3), 1975, pp.~179--196


\bibitem{serfozo82} {\sc R. Serfozo}, {\em Convergence of Lebesgue Integrals with Varying Measures}, Sankhy{\=a}: The Indian J. Stat., Series A, 44(3), 1982, pp.~380--402.


\bibitem{tzortzis2015dynamic} {\sc I. Tzortzis and C.D. Charalambous and T. Charalambous}, {\em Dynamic Programming Subject to Total Variation Distance Ambiguity}, SIAM J. Control Optim., 53(4), 2015, pp.~2040--2075.
 
\bibitem{VanHandel} {\sc R.~van Handel}, {\em The Universal Glivenko--Cantelli Property}, Probab. Theory Related Fields, 155(3--4), 2012, pp.~911--934.

\bibitem{Vap00} {\sc V. N. Vapnik}, {\em The Nature of Statistical Learning Theory}, 2nd edition, Springer, New York, 2000. 

\bibitem{villani2008optimal} {\sc C. Villani}, {\em Optimal Transport: Old and New}, Springer, 2008.
 
\bibitem{wheeden77} {\sc R. L. Wheeden and A. Zygmund}, {\em Measure and Integral}, Marcel Dekker, New York, 1977.

\bibitem{whittle1991risk} {\sc P. Whittle}, {\em A Risk-Sensitive Maximum Principle: The Case of Imperfect State Observation},  IEEE Trans. Automat. Control, 36(7), 1991, pp.~793--801.
 
\bibitem{WuVer11} {\sc Y. Wu and S. Verd\'u}, {\em Witsenhausen's Counterexample: A View from Optimal Transport Theory}, Proc. IEEE Conference on Decision and Control, 2016, pp.~5732--5737.

\bibitem{WuVerdu} {\sc Y. Wu and S. Verd\'u}, {\em Functional Properties of Minimum Mean-Square Error and Mutual Information}, IEEE Trans. Inform. Theory, 58(3), 2012, pp.~1289--1301.
 
\bibitem{yuksel12:siam} {\sc S. Y\"uksel and T. Linder}, {\em Optimization and Convergence of Observation Channels in Stochastic Control}, SIAM J. Control Optim., 50, 2012, pp.~864--887.

\bibitem{Yus76} {\sc A.A.~Yushkevich}, {\em Reduction of A Controlled Markov Model with Incomplete Data to A Problem with Complete Information in the Case of Borel State and Control Spaces}, Theory Prob. Appl., 21, 1976, pp.~153--158.

\bibitem{zames1981feedback} {\sc G. Zames}, {\em Feedback and Optimal Sensitivity: Model Reference Transformations, Multiplicative Seminorms and Approximate Inverses}, IEEE Trans. Automat. Control, 26(2), 1981, pp.~301--320.

\bibitem{zhou1996robust} {\sc K. Zhou, J. C. Doyle and K. Glover}, {\em Robust and Optimal Control}, Volume 40, Prentice-Hall, New Jersey 1996.

\bibitem{zhu2015risk} {\sc H. Zhu and E. Zhou}, {\em Risk Quantification in Stochastic Simulation under Input Uncertainty},  arXiv preprint arXiv:1507.06015, 2015.
  

\end{thebibliography}

\end{document}